\documentclass[phd,project]{ppgccufmg}    
\usepackage[english]{babel}      
\usepackage[latin1]{inputenc}
\usepackage[T1]{fontenc}
\usepackage{type1ec}
\usepackage{graphicx}

\usepackage[a4paper,
  english,
  bookmarks=true,
  bookmarksnumbered=true,
  linktocpage,
  colorlinks,
  citecolor=black,
  urlcolor=black,
  linkcolor=black,
  filecolor=black,
  ]{hyperref}
\usepackage[square]{natbib}
\usepackage{graphicx}
\usepackage{caption}
\usepackage{subcaption}
\usepackage{booktabs}
\usepackage{amsthm}
\usepackage{verbatim}
\usepackage{amsmath}
\usepackage{setspace}
\newtheorem{defn}{\textbf{Definition}}
\newtheorem{theorem}{\textbf{Theorem}}
\newtheorem{lemma}{\textbf{Lemma}}

\usepackage{algpseudocode,algorithm}
\usepackage{mathtools}

\algdef{SE}[SUBALG]{Indent}{EndIndent}{}{\algorithmicend\ }%
\algtext*{Indent}
\algtext*{EndIndent}
\newcommand\tab[1][0.8cm]{\hspace*{#1}}
\newcommand\tabs[1][0.3cm]{\hspace*{#1}}

\usepackage{todonotes}
\usepackage{multicol,caption}

\algnewcommand\algorithmicforeach{\textbf{for each}}
\algdef{S}[FOR]{ForEach}[1]{\algorithmicforeach\ #1\ \algorithmicdo}

\begin{document}

\ppgccufmg{
  title={Concurrent Self-Adjusting Distributed Tree Networks},
  authorrev={Soares Peres, Bruna},
  cutter={D1234p},
  cdu={519.6*82.10},
  university={Federal University of Minas Gerais},
  course={Computer Science},
  portuguesetitle={Topologia Auto-Ajust\'avel Concorrente e Distribu\'ida para Comunica\c c\~ao Par-a-Par em Redes de Computadores},
  portugueseuniversity={Universidade Federal de Minas Gerais},
  portuguesecourse={Ci\^encia da Computa\c c\~ao},
  address={Belo Horizonte},
  date={2017-05},
  advisor={Olga Goussevskaia},
  abstract=[brazil]{Resumo}{resumo},
  abstract={Abstract}{abstract},
}


\chapter{Introduction}\label{chap:intro}

\section{Motivation}
The popularity of cloud-based networks, virtualized desktops and servers, and remote data-storage devices are increasing and organizations are taking advantage of these technologies. Network virtualization is the key to the success of cloud computing, and SDN is a really helpful technology in network virtualization~\cite{Jain13}. Software defined networking (SDN) has emerged as an efficient technology to support the dynamic of future networks in an efficient way~\cite{Sezer13}. While Cloud systems are usually hosted in large data-centers and are centrally managed, some works have proposed other types of Cloud architectures. Babaoglu et al. ~\cite{Babaoglu12}, for example, presents a fully decentralized P2P Cloud, that  allows organizations or even individual to build a computing infrastructure out of existing resources, which can be easily allocated among different tasks. Other applications, such as multimedia cloud computing, can also benefit from P2P architecture. In multimedia cloud computing, users can store and process in a distributed manner their multimedia application data in the cloud. P2P multimedia computing refers to a distributed application architecture that partitions multimedia-computing tasks or workloads between peers~\cite{Zhu11}.

All of these technologies have something in common: the network topology can be represented by a distributed data structure and networks are optimized toward static metrics, such as the diameter or the length of the longest route. Furthermore, communication between two nodes in the network is performed through requests. The requests flow suffers many changes over time, increasing the use of different data paths at different times. Since the network topology is built in advance, given static metrics, it does not react to the real time data traffic demand.

The design of scalable and robust overlay topologies has been a main research subject since the very origins of peer-to-peer (P2P) computing. In \cite{AvinBS13}, Avint et al. initiate the study of topologies optimized to serve the communication demand. The authors propose a simple overlay topology composed of $k$ (rooted and directed) Binary Search Trees (BSTs), based on splay tree concepts \cite{Sleator:1985:SBS:3828.3835}. In \cite{Avin13} and \cite{Schmid15}, the authors present SplayNet, a self-adjusting distributed tree network that improves the communication cost between two nodes. In SplayNets, nodes communicating more frequently should become topologically closer to each other. Unlike in classical splay trees, where requests always originate from the root (which however can change over time), in SplayNets communication happens between arbitrary node pairs in the network. In contrast to these flexible classic data structures, and with the exception of SplayNets, today's distributed data structures and networks are still optimized toward static metrics, such as the diameter or the longest route~\cite{conf/sigcomm/SchlinkerMSMVYK15}.

\section{Objectives}

SplayNets are a generalization of splay trees in computer networks. Splay Trees are a self-adjusting form of binary search tree, first introduced by Selator and Tarjan \cite{Sleator:1985:SBS:3828.3835}. In such structures, frequently accessed elements are moved closer to the root to improve the average access time given the element's popularity. Unlike in classical splay trees, where requests always originate from the root, in SplayNet communication happens between arbitrary node pairs in the network. Communication between two nodes in the network is performed through requests, and requests flows suffer many changes over time, increasing the use of different data paths at different times.

There are several applications for concurrent SplayNets, such as self-adjusting peer-to-peer overlays, but also in datacenter environments, where emerging technologies, like the ProjecToR interconnect~\cite{projector}, allow top-of-the-rack switches to establish direct links over time. ProjecToR is a novel approach that enables reconfigurable data center interconnection. It uses a digital micromirror device (DMD) and mirror assembly combination as a transmitter and a photodetector on top of the rack as a receiver, instead of electrical packet switches interconnecting racks in a multi-tier topology using optical fiber cables. This allows all racks to establish direct links with a fast reconfiguration, i. e., reconfiguring a link takes only 12 microseconds, a desirable feature for SplayNet.

The analysis presented in \cite{Roy:2015} shows that, unlike reported on in literature the majority of traffic in a data center is not rack-local. Instead, almost 60\% of all traffic from Facebook's data center clusters are inter-rack (intra-cluster). Thus, combining SplayNet with ProjecToR can optimize the communication performance between racks inside a data center.

Our objective is to implement a distributed and concurrent locally self-adjusting tree network, and analyze how concurrent operations fundamentally affect the algorithm's performance. In addition, we want to evaluate the SplayNet performance with different workloads through simulations. Finally, we intend to understand better how SplayNet behaves by bringing together analysis and existing models of real workloads.

\section{Contribution}
While SplayNets are inherently intended to distributed applications, so far, only sequential algorithms are known to maintain SplayNets. In this work, we present the first distributed and concurrent implementation of SplayNets.

We show safety and liveness of our algorithm, and we analyze the work of the algorithm. In particular, we prove that it is loop- and deadlock-free, and we show that the amortized average cost is increased by only a logarithmic factor, when compared to the non-concurrent scenario, analyzed in~\cite{Schmid15}. 

Moving from centralized to distributed and concurrent algorithms is a challenging task, due to additional complexity needed to keep the network in a consistent state. But it also brings benefits, by simplifying costly tasks, like global request scheduling and network expansion, in particular, while the network continues to carry traffic~\cite{conf/sigcomm/SchlinkerMSMVYK15}. 

\section{Completed Tasks}
In this Section we present the completed tasks of the Thesis Project proposed. A brief announcement submission of this works was submitted to the International Symposium on Distributed Computing (DISC 2017).

We first present a design overview of the distributed implementation of SplayNet. We implement concurrency control using (conservative) read-write locks on the (bidirectional) 
links, maintained by each node in the network. To prevent nodes from starving, every node maintains a local buffer with a queue of rotation requests received from its neighborhood. The buffer follows a priority routine that defines a global order in the network.

Analytical results show that our proposed algorithm prevents loops and deadlocks from occurring between concurrent rotations. We show that the duration of a round, i.e., the time between a node requesting and completing a rotation, is $O(logm)$, where $m$ is the number of concurrent splay requests in the network. Finally, we compute the total amortized average cost of a splay request in number of rounds and number of time-slots and as a function of the empirical entropies of source and destination nodes of the splay requests.


\section{Organization}

The remainder of this work is divided as follows. In Chapter \ref{chap:related} we present the related work. In Chapter \ref{chap:model} we introduce briefly the model of SplayNet proposed in \cite{Avin13} and \cite{Schmid15}. In Chapter \ref{chap:design} we present our SplayNet design overview, considering concurrent splay operations. In Chapter \ref{chap:analysis} we analyze the cost of SplayNet in a concurrent scenario. 
Finally, in Chapter \ref{chap:future}, we discuss the future directions of this thesis project.

\chapter{Related Work}\label{chap:related}


The challenge of designing an efficient and robust overlay topology for peer-to-peer computing has been a research subject since the origins of such technology. Initially, the design of peer-to-peer (p2p) topologies considers the optimization of static properties. In \cite{AvinBS13}, Avint et al. initiate the study of topologies optimized to serve the communication demand. The authors propose a simple overlay topology composed of $k$ (rooted and directed) Binary Search Trees (BSTs), based on splay tree concepts\cite{Sleator:1985:SBS:3828.3835}.

In \cite{Avin13} and \cite{Schmid15}, the authors present SplayNet, a self-adjusting distributed tree network that improves the communication cost between two nodes. In SplayNet, nodes communicating more frequently should become topologically closer to each other. In contrast to these flexible classic data structures, today's distributed data structures and networks are still optimized toward static metrics, such as the diameter or the length of the longest route. Furthermore, they analyze the performance of SplayNet (in terms of amortized costs) and show that the overall cost is upper bounded by the empirical entropies of the sources and destinations in the communication pattern. A simple lower bound follows from conditional empirical entropies. They also prove the optimality of SplayNet in specific case studies, e.g., when the communication pattern follows a product distribution.

Self-adjusting networks have many applications, ranging from self-optimizing peer-to-peer topologies over green computing (e.g., due to reduced energy consumption) \cite{Heller10} to adaptive virtual machine migrations \cite{Arora11}, \cite{Shang10}, microprocessor memory architectures \cite{Lis11}, and grids \cite{Batista07}. Other self-adjusting routing scheme were considered, e.g., in scale-free networks to overcome congestion \cite{Tang09}. Peer-to-peer networks are particularly interesting dynamic systems as they are very transient and members continuously join and leave. In this sense, a peer-to-peer system can never be fully repaired but must always be fully functional. Today, peer-to-peer networking is a relatively mature field of research, and there are many solutions to maintain desirable network properties under both randomized \cite{Scheideler2009} as well as worst case \cite{Kuhn2010} membership changes, and some peer-to-peer networks are even self-stabilizing \cite{JRSST-PTADSSG-09} in the sense that they quickly converge to a desirable topology (e.g., a hypercube) from an arbitrary connected structure. However, none of these systems are self-adjusting to the demand.

SplayNet is built upon classic literature on self-adapting data structures, in particular upon the seminal work of Sleator and Tarjan on splay trees \cite{Sleator:1985:SBS:3828.3835}: Splay trees are optimized binary search trees which move more popular items closer to the root in order to reduce the average access time. Splay trees and its variants (e.g., Tango trees \cite{Demaine04} or multi-splay trees \cite{Wang:2006}) have been studied intensively for many years (e.g. \cite{Allen78}, \cite{Sleator:1985:SBS:3828.3835}), and the famous dynamic optimality conjecture continues to puzzle researchers: The conjecture claims that splay trees perform as well as any other binary search tree  algorithm \cite{Demaine04}, \cite{Wang:2006}. In contrast to the classical splay tree data structures, on SplayNet the lookups or requests cannot only originate from a single root, but communication happens between all pairs of nodes in the network. Hence, SplayNet is in the realm of distributed data structures or networking.

The Flattening algorithm, presented in \cite{Reiter08}, is a distributed algorithm that improves the performance of accessing one node from another node in a k-ary tree by adjusting the tree as nodes communicate with each other. In contrast to SplayNet, Flattening does not require preserving the order of nodes in the tree. However, this ordering is an attractive characteristic since it allows locally routing: given a destination identifier (or address), each node can decide locally whether to forward the packet to its left child, its right child, or its parent (in a binary tree, for example). The total actual work done by a sequence of $m$ Flattening operations is at most $3m + (2m+n)\log n$.

CB Tree, presented in \cite{Afek2014} is a concurrent self-adjusting binary search tree inspired by splay trees. As in splay trees, more-frequently accessed items move closer to the root. Instead of using the classical rotations (zig, zig-zig and zig-zag), they divide the splay into single and double rotations. In addition, the authors propose a counting-based method, in which each operation in a CB tree is allowed to do single and double rotations at nodes on the access path only when such a single or double rotation reduces the tree potential by at least $\delta$, where $\delta$ is a positive constant $< 2$. However, the amortized bounds on the path length and number of rotations are calculated only for the sequential algorithm, and a more detailed analysis of how concurrent operations fundamentally affect the algorithm's performance is not presented.
\chapter{Model and Background}\label{chap:model}
    \begin{defn}\label{def:splaynet}\textbf{Network model:} : A SplayNet $\mathcal{T}$ is comprised of a set of $n$ communication nodes with distinct identifiers, interacting according to a certain communication pattern $\mathcal{F}$. Differently from \cite{Schmid15}, where the communication pattern was modeled as a sequence of consecutive communication requests, we consider a set of $m$ concurrent communication requests. The goal is to dynamically find a locally routable binary search tree (BST) topology, which connects all nodes and optimizes the routing cost for $\mathcal{F}$, making local topology transformations, called rotations, before each request is served.
    \end{defn}
    
    \begin{defn}\label{def:model} \textbf{Communication model}: nodes communicate via reliable and synchronous message passing: nodes do not fail and the execution is partitioned into time-slots;
    \end{defn}
    
    \begin{defn} \textbf{Lowest common ancestor $\alpha(a,b)$}: the lowest common ancestor of two nodes $a$ and $b$ ($\alpha(a,b)$) in $\mathcal{T}$ is the lowest node that has both $a$ and $b$ as descendants. A node can be the lowest common ancestor of itself and another node.
    \end{defn}
    
    \begin{defn} \textbf{Set of concurrent splay requests $\mathcal{F}$}: Given $n$ nodes distributes in a binary tree $\mathcal{T}$, splay requests can occur concurrently in $\mathcal{T}$. We call $\mathcal{F}_{\mathcal{R}}$ the set of concurrent splay requests in super-round $\mathcal{R}$. In addition, $\mu_{\mathcal{R}}$ and $\delta_{\mathcal{R}}$ represent, respectively, the set of source and destination nodes of a splay in super-round $\mathcal{R}$. An element $f \in \mathcal{F}_{\mathcal{R}}$, is a splay request of an oriented pair of source and destination $(s,d)$, $\mathcal{S}(s,d)$, $\mid s \in \mu_{\mathcal{R}}$ and $d \in \delta_{\mathcal{R}}$. When clear from the context, we will often omit the super-round $\mathcal{R}$ and simply write $\mathcal{F}$.
    \end{defn}
    
    \begin{defn}\label{def:rotation}\textbf{Rotation $\beta(u)$}: Rotations are local transformations of tree network, performed atomically. A rotation to move a node $u$ up depends upon the relative positions of $u$, its parent $v$ and its grandparent $w$. Therefore, there are three different rotation types:

    \begin{itemize}
        \item \textbf{zig:} $u$'s grandparent does not participate in a zig rotation ($u$ may not have a grandparent), see Figure \ref{fig:zig}. In this case, we rotate $u$ over $v$, making $u$'s children be node $v$ and a previous subtree $T_1$, keeping the subtrees intact. Algorithm \ref{alg:zig} presents the sequential centralized algorithm for a zig rotation.
        
        \item \textbf{zig-zig:} $u$ and its parent $v$ are both left or right children, (see figure \ref{fig:zigzig}). $w$ is replaced by $u$, $v$ becomes a child of $u$ and $w$ a child of $v$, keeping the subtrees intact. Algorithm \ref{alg:zigzig} presents the sequential centralized algorithm for a zig-zig rotation.
        
        \item \textbf{zig-zag:} $u \oplus v$ is a left child and the other one is a right child, see Figure \ref{fig:zigzag}. The rotation replaces $w$ by $u$ and $u's$ children become $v$ and $w$, keeping the subtrees intact. Algorithm \ref{alg:zigzag} presents the sequential centralized algorithm for a zig-zag rotation.
    \end{itemize}

\begin{figure}
	\centering
	\begin{subfigure}{0.4\textwidth} 
		\includegraphics[width=\textwidth]{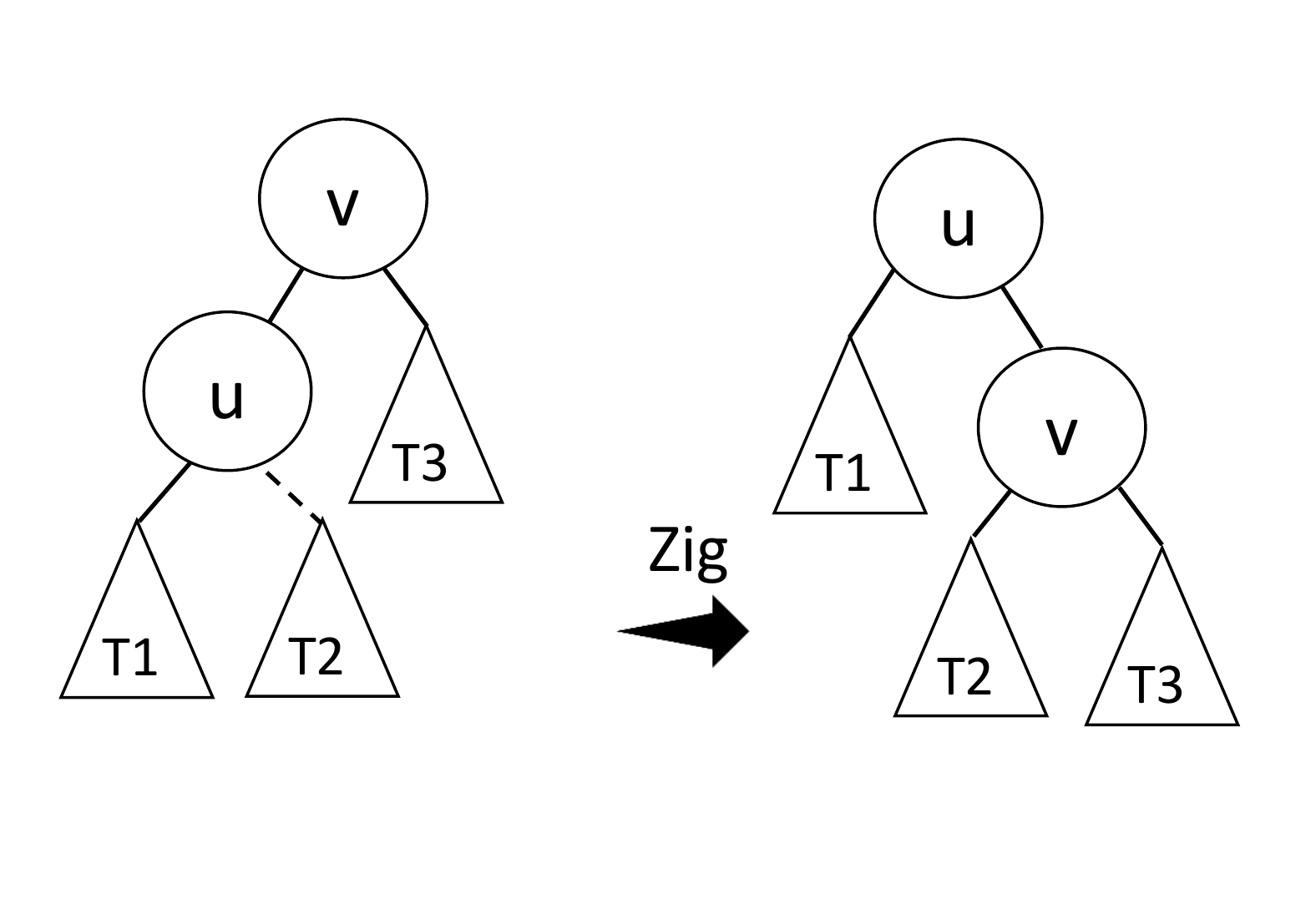}
		\caption{}
		\label{fig:zig}
	\end{subfigure}
	\vspace{1em} 
	\begin{subfigure}{0.4\textwidth} 
		\includegraphics[width=\textwidth]{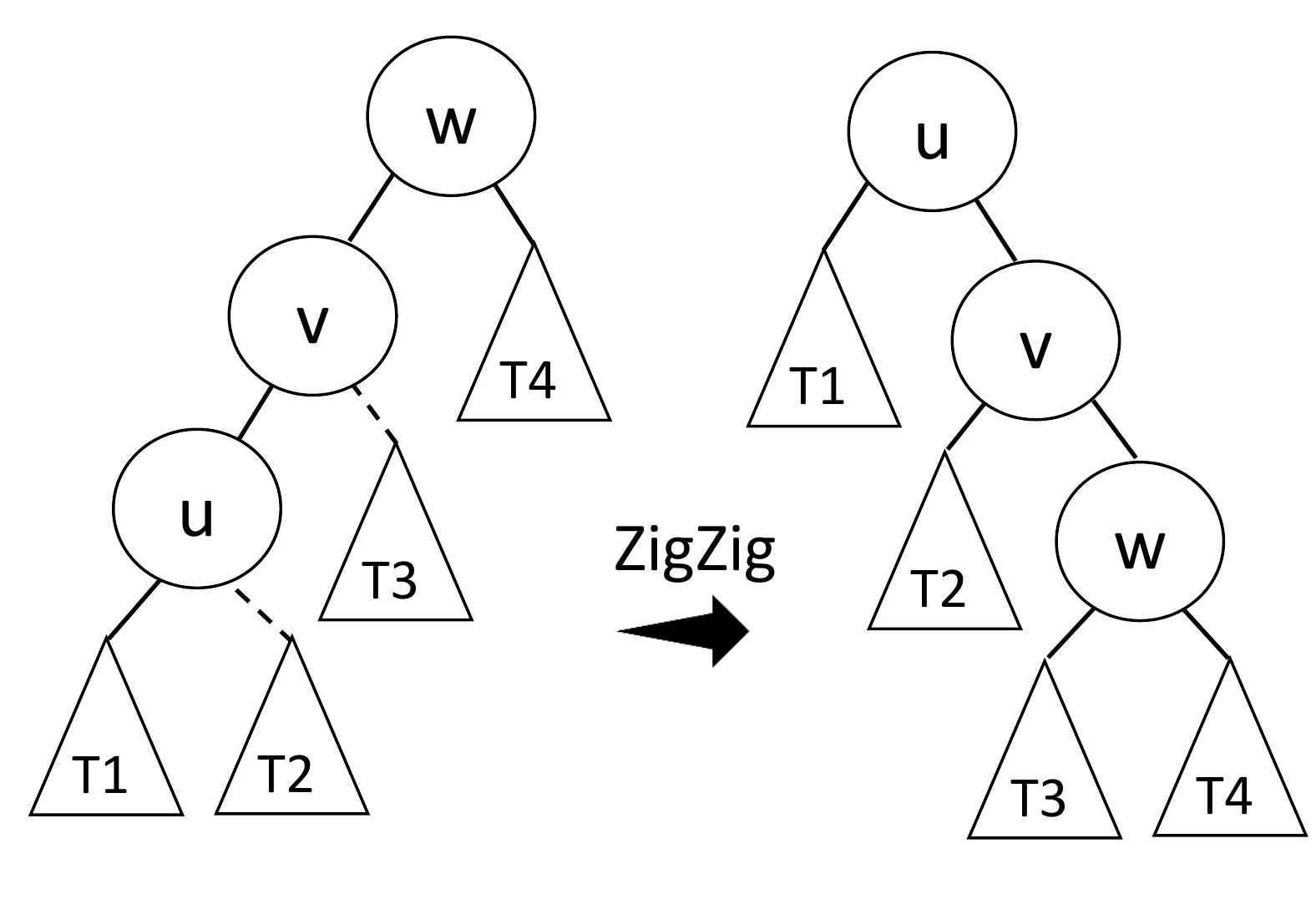}
		\caption{}
		\label{fig:zigzig}
	\end{subfigure}	
	\begin{subfigure}{0.4\textwidth} 
		\includegraphics[width=\textwidth]{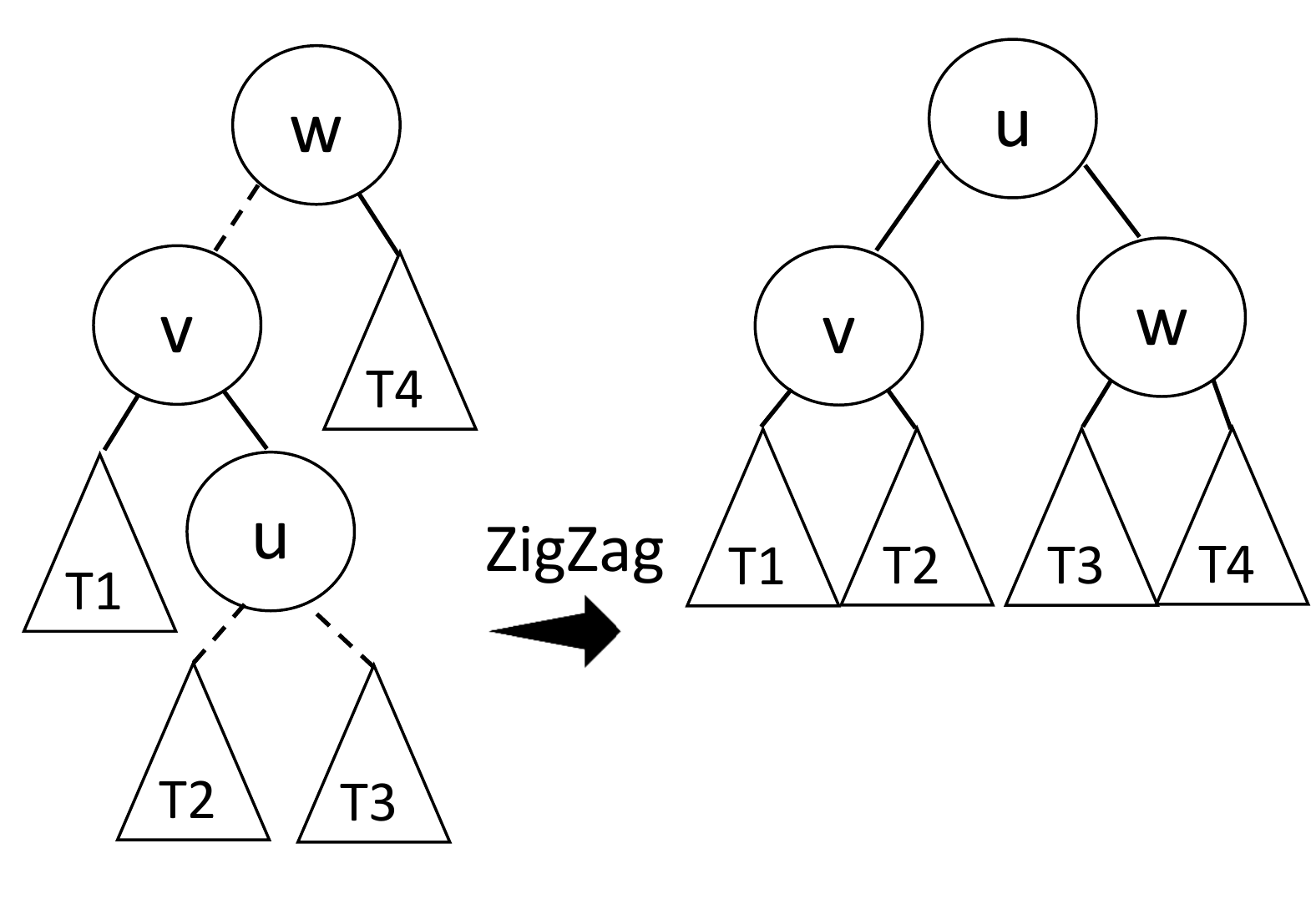}
		\caption{}
		\label{fig:zigzag}
	\end{subfigure}	
	
    \caption{Rotations: (a) Zig, (b) ZigZig and (c) ZigZag. The dashed lines indicate adjacency relationships which are not maintained during the operation}\label{fig:rotations}
\end{figure}

    \end{defn}

    \begin{algorithm}[!ht]
    \caption{Zig Algorithm}\label{alg:zig}
    \begin{algorithmic}[1]
    \Function{$\beta$}{${u}$}
        \If {u > v}
            \State v.r := u.l
            \State u.l.p := v
            \State u.l := v 
            \State v'' := v.r''
            \State u' := v'
        \Else
            \State v.l := u.r
            \State u.r.p := v
            \State u.r := v
            \State v' := v.l'
            \State u'' := v''
        \EndIf
        \If {v > v.p}
            \State (v.p).r := u
        \Else
            \State (v.p).l := u
        \EndIf
        \State u.p := v.p
        \State v.p := u
    \EndFunction
    \end{algorithmic}
    \end{algorithm}
    
\begin{figure*}[ttt!]
 \begin{minipage}[t]{3.15in}
 \begin{algorithm}[H]
    \caption{Zig-Zig Algorithm}\label{alg:zigzig}
    \begin{algorithmic}[1]
    \Function {$\beta$}{${u}$}
        \If {u > w}
            \State w.r := v.l
            \State v.l.p := w
            \State v.l := w
            \State v.r := u.l
            \State u.l.p := v
            \State u.l := v
            \State w'' := w.r''
            \State v' := w'
            \State v'' := v.r''
            \State u' := w'
        \Else
            \State w.l := v.r
            \State v.r.p := w
            \State v.r := w
            \State v.l := u.r
            \State u.r.p := v
            \State u.r := v
            \State w' := w.l'
            \State v' := v.l'
            \State v'' := w''
            \State u'' := w''
        \EndIf
        \If {w > w.p}
            \State (w.p).r := u
        \Else
            \State (w.p).l := u
        \EndIf
        \State u.p := w.p
        \State v.p := u
        \State w.p := v
    \EndFunction
    \end{algorithmic}
    \end{algorithm}
 \end{minipage}
 \hfill
 \begin{minipage}[t]{3.15in}
\begin{algorithm}[H]
    \caption{Zig-Zag Algorithm}\label{alg:zigzag}
    \begin{algorithmic}[1]
    \Function {$\beta$}{${u}$}
        \If {u > w}
            \State w.r := u.l
            \State u.l.p := w
            \State v.l := u.r
            \State u.r.p := v
            \State u.r := v
            \State u.l := w
            \State w'' := w.r''
            \State v' := v.l'
            \State u' := w'
            \State u'' := v''
        \Else
            \State w.l := u.r
            \State u.r.p := w
            \State v.r := u.l
            \State u.l.p := v
            \State u.r := w
            \State u.l := v
            \State w' := w.l'
            \State v'' := v.r''
            \State u' := v'
            \State u'' := w''
        \EndIf
        \If {w > w.p}
            \State (w.p).r := u
        \Else
            \State (w.p).l := u
        \EndIf
        \State u.p := w.p
        \State v.p := u
        \State w.p := u
    \EndFunction
    \end{algorithmic}
    \end{algorithm}
 \end{minipage}
 \hfill
\end{figure*}

    \begin{defn} \textbf{Splay $\mathcal{S}(s,d)$}: A splay with source node $s\in\mathcal{T}$ and destination node $d\in \mathcal{T}$ is an ordered set of rotations performed in parallel by $s$ and $d$, such that there are two time instances $\tau_s$ and $\tau_d$, such that, $s = \alpha(s,d)$ in $\tau_s$ and $d(s,d)=1$ in $\tau_d$ or $d = \alpha(s,d)$ in $\tau_d$ and $d(s,d)=1$ in $\tau_s$.
    \end{defn}
 
    \begin{defn} \textbf{Path set $\mathcal{P}_{\tau}(s,d)$}: path between communication nodes $s$ and $d$, $s \in \mu, d \in \delta$ in $T_{\tau}$. $|\mathcal{P}_{\tau}| = d_{\tau}(s,\alpha(s,d)) + d_{\tau}(d,\alpha(s,d))$, where $d_{\tau}(a,b)$ is the distance between two nodes $a$ and $b$ in $\mathcal{T}$ at time-slot $\tau$.
    \end{defn}
    
    \begin{defn}\label{def:OF} \textbf{Splay Objective}: We say that a splay request $\mathcal{S}(s,d)$ has achieved its objective when $d(s,d)=1$, which means that $s = d.p$ or $d = s.p$.    
    \end{defn}
    
    \begin{defn}\label{def:buffer} \textbf{Buffer $\mathcal{B}$}: a queue of size $|\mathcal{B}|$, that stores rotation requests from a node's children and grandchildren, and its own. Each node maintains a buffer, which must be updated when a rotation is performed by one of the node's neighbors.
    \end{defn}
    
    \begin{defn} \textbf{Carried child}\label{def:cchild}: Given a rotation $\beta(u)$, if a node $x$ is a child of $u$ before and after $\beta(u)$, then $x$ is a \textit{carried child} by that rotation. For example, in Figure \ref{fig:zig}, the root of sub-tree $T_1$, say $t_1$, is a carried child by $\beta(u)$.
    \end{defn}
    
    \begin{defn} \textbf{Abandoned child}\label{def:achild}: Given a rotation $\beta(u)$, if a node $x$ is a child of $u$ before \textit{but not} after $\beta(u)$, then $x$ is an \textit{abandoned child} by that rotation. For example, in Figure \ref{fig:zig}, the root of sub-tree $T_2$, say $t_2$, is an abandoned child by $\beta(u)$.
    \end{defn}    
    
    \textbf{Synchronization:} While we prove our algorithms correct (in terms of safety and liveness) for an arbitrary asynchronous model (without failures), for the complexity analysis, we we consider a simplified synchronous model, where time is partitioned into synchronous time-slots, and every node can reliably exchange a constant number of messages with its neighbors. Moreover, each node keeps track of the number of splay requests and the number of rotations it has performed. It considers itself to be in round $t$ and super-round $\mathcal{R}$ after completing a rotation with sequence number $t$ within a splay request with sequence number $\mathcal{R}$.

    \begin{defn} \textbf{Time-slot $\tau$}: a (globally synchronized) period of time, in which every node can send a message to each neighbor, the messages are delivered, and every node performs a local computation;
    \end{defn}
    
    \begin{defn}\label{def:round} \textbf{Round $t$}: a period of time between a node requesting a rotation and completing it. The length of a round is the number of time-slots it takes a node to perform one rotation (we provide an upper bound on the length of a round in Lemma~\ref{claim:round}). Note that, in a given time-slot $\tau$, different nodes can be in different rounds.
    \end{defn}
    
    \begin{defn} \textbf{Super-round $\mathcal{R}$}: a period of time in which each node can request and complete one splay operation. The length of a super-round $|\mathcal{R}|$ is the number of rounds it takes for every source-destination pair in $\mathcal{T}$ to complete one splay request (we provide an upper bound on the length of a super-round in Theorem~\ref{thm:totalcost}).
    \end{defn}

Differently from \cite{Schmid15}, where the communication pattern was modeled as a sequence of consecutive communication requests, we consider a set of concurrent communication requests. Note that, in our model, each communication request starts a splay request.

\begin{defn} \textbf{Set of concurrent splays requests $\mathcal{F}$}: Given a SplayNet $\mathcal{T}$ on $n$ nodes and a super-round $\mathcal{R}$, $\mathcal{F}_{\mathcal{R}}$, $|\mathcal{F}_{\mathcal{R}}|=m$, denotes the set of $m$ concurrent splay requests in $\mathcal{R}$. In addition, let $\mathcal{S}_{\mathcal{R}} \subseteq \mathcal{T}$ and $\mathcal{D}_{\mathcal{R}}\subseteq \mathcal{T}$ represent, respectively, the sets of sources and destinations in $\mathcal{F}_{\mathcal{R}}$, i.e., $\mathcal{S}(s,d) \in \mathcal{F}_{\mathcal{R}}$, $s\in \mathcal{S}_{\mathcal{F}}$, $d\in \mathcal{D}_{\mathcal{F}}$. When clear from the context, we will omit the super-round $\mathcal{R}$ and simply write $\mathcal{F}$.
\end{defn}

\textbf{Algorithm:} The SplayNet algorithm presented in \cite{Schmid15} is a natural generalization of the classic splay tree algorithm. It is based on a double splay strategy: similarly to classic splay trees, SplayNet aggressively moves communicating nodes together; however, rather than splaying nodes to the root of the BST, locality is preserved in the sense that the source and the destination node are only rotated to their common ancestor.

Concretely, consider a communication request $(s,d)$ from node $s$ to node $d$ in time-slot $\tau_i$, and let $\alpha_i(s, d)$ denote the lowest common ancestor of $s$ and $d$ in the current network $\mathcal{T}_i$. Unlike the algorithm proposed in \cite{Schmid15}, in which first $s$ splays towards to the lowest common ancestor, and then $d$ splays until $s = d.p$, we propose a different approach, in which $s$ and $d$ rotate in parallel until $d(s,d)=1$. Thus, when a request $(s,d)$ occurs, $s$ and $d$ splay towards to the lowest common ancestor $\alpha_{i}(s, d)$, using the classic splay operations zig, zig-zig and zig-zag (Figure \ref{fig:rotations}).
Once $s=\alpha_{j}(s,d)$ or $d=\alpha_{j}(s,d)$ in a time-slot $t_{j}\geq t_{i}$, the lowest common ancestor ($s$ or $d$) waits until the other node in the splay becomes its child. If this condition of being the lowest common ancestor changes, due to another rotation, then, the node resumes its rotations.


\noindent\textbf{Entropy:} A useful parameter when evaluating the performance of SplayNets is the \textit{empirical entropy} of the communication request set $\mathcal{F}$ , i.e., the entropy implied by the communication frequencies.
\begin{defn}\label{def:entropy}
Let $\hat{X}(\mathcal{F})=\{f(x_1),\ldots, f(x_n)\}, x_i \in \mathcal{S}_{\mathcal{F}}$ (or, analogously, $\hat{Y}$) be the frequency distribution of the communication sources (or destinations) in the communication sequence $\mathcal{F}$, i.e., $f(x_i) =(\#x_i (\text{ or } y_i)$ is a source (or destination) in $\mathcal{F})/m$. The empirical entropy is defined as follows, for communication sources (and destinations):

\begin{equation}
H(\hat{X})=\sum_{i=1}^n{f(x_i)\log_2{\frac{1}{f(x_i)}}}
\end{equation}
\begin{equation}
H(\hat{Y})=\sum_{i=1}^n{f(y_i)\log_2{\frac{1}{f(y_i)}}}
\end{equation}
\end{defn}

In \cite{Schmid15}, it was shown that the amortized average cost per splay request is $O(H(\hat{X}) + H(\hat{Y}))$. In \cite{Avin13}, the authors assume that $\mathcal{F}$ is a sequence of consecutive splay requests and compute an upper bound on the amortized communication cost as a function of the entropy of the sources and destinations of the requests:
\begin{theorem}\label{thm:entropies}\cite{Avin13} Let $\mathcal{F}$ be an arbitrary sequence of (non-concurrent) communication requests, then for any initial tree $\mathcal{T}_{0}$,
\begin{equation}
Cost(SplayNet,\mathcal{T}_{0},\mathcal{F}) = O(H(\hat{X}) + H(\hat{Y})).\nonumber
\end{equation}
\end{theorem}
\chapter{Design Overview}\label{chap:design}

In this chapter, we present an overview of the distributed implementation of SplayNet. We cover issues, such as concurrency locks (Section \ref{subsec:lock}), request priorities (Section \ref{subsec:buffer}) and distributed routines (Section \ref{subsec:distributed}).

\section{Concurrency locks}\label{subsec:lock}


In order to allow concurrent rotations, while maintaining a consistent SplayNet topology, some local variables must be synchronized. Unlike proposed in \cite{Reiter08}, in which nodes implement mutually exclusive access to a shared token, we want nodes to perform rotations concurrently, without breaking consistency. 
Consider a rotation $\beta_{\tau}(u)$ and nodes $u, v, w$ and $z \mid v = u.p$, $w= v.p$ and $z = w.p$. In order to guarantee consistency, the following nodes need to be synchronized, or locked, during the rotation, according to rotation's type (see Figure \ref{fig:rotations}:

\begin{enumerate}
\item $\beta_{\tau}(u)$ is a zig: $w$, $v$ and $u$. 

\item $\beta_{\tau}(u)$ is a zig-zig or zig-zag: $z$, $w$, $v$ and $u$.
\end{enumerate}

Even though the link between nodes $u$ and $u.r$ changes because of a rotation, there is no need to lock $u.r$. Since $u$ is locked, another concurrent rotation $\beta(x)$ cannot change the link between $u$ and $u.r$ concurrently with $\beta(u)$. Therefore, if $u.r$ is locked for a concurrent rotation $\beta(x)$, $u.r$ can only be the grand-grandparent in that rotation (or the grandparent if $\beta(x)$ is a zig). The reason of this is that the only link of $u.r$ that $\beta(u)$ changes is $u.r.p$, and the only link changed in the grand-grandparent of a rotation is one of the links to its children ($u.r.r$ or $u.r.l$). This argument applies to $u.l$ and $v.r \lor v.l \neq u$.

Figure \ref{fig:lockingexample} shows an example of local lock for a rotation $\beta(u)$. Nodes in gray ($u$,$v$,$w$) are those \textit{participating} in the rotation: $u$ is the level-one node, $v$ is the level-two node and $w$ is the level-three node. The dotted node ($z$) is just \textit{locked} due to a link change, but does not effectively participate in the rotation. All the nodes that have a link change due to a rotation but do not \textit{participate} in the rotation (in this example $z$, $t3$ and $t2$) are notified of the link changes by the nodes participating in the rotation right after it has occurred.

 \begin{figure}[H]
    \centering
    \includegraphics[width=0.4\linewidth]{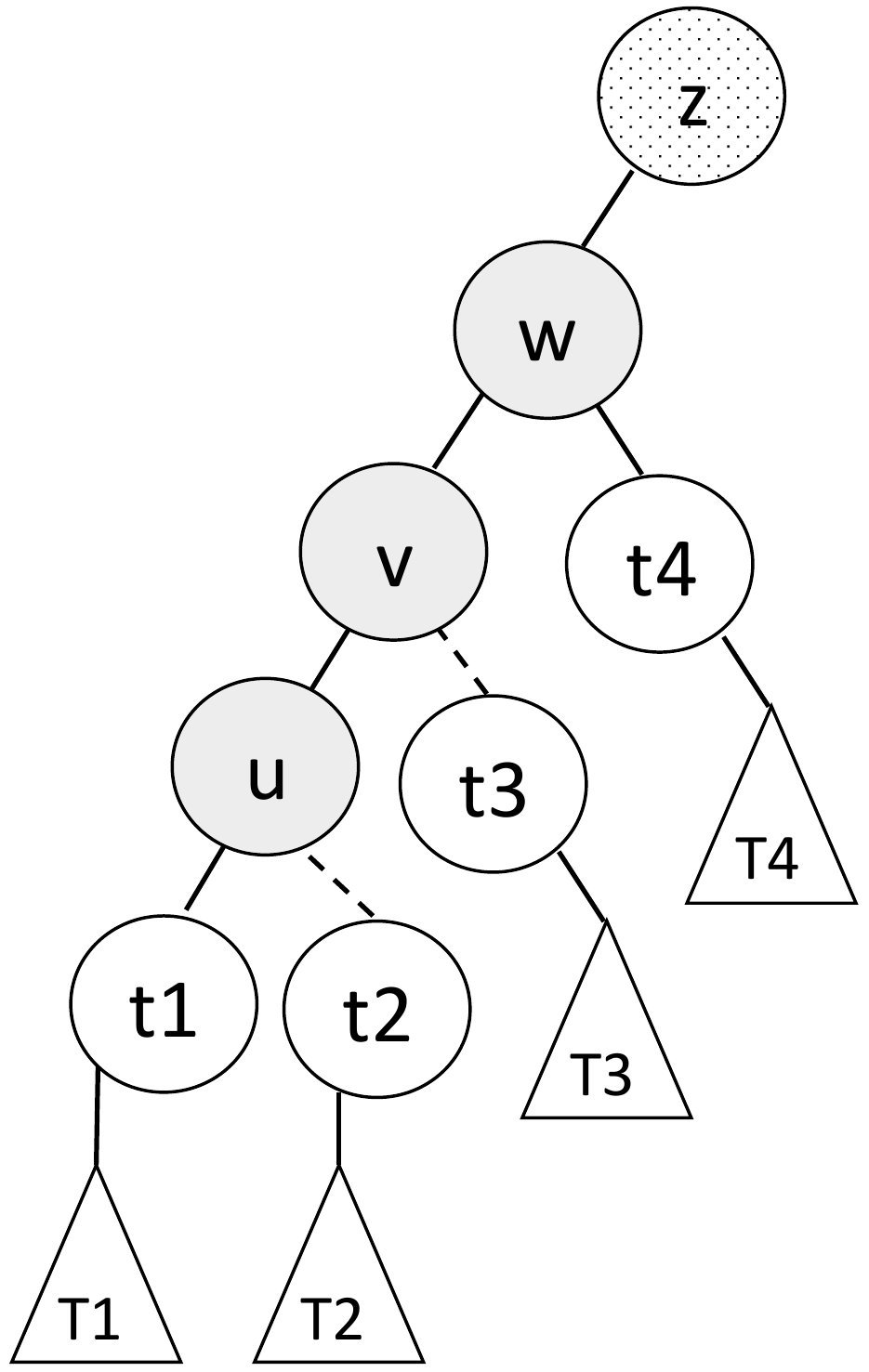}
    \caption{Local locks for rotation $\beta(u)$ (nodes in gray are locked and \textit{participate} in $\beta(u)$; dotted nodes are just \textit{locked}, i.e., have a link update).}
    \label{fig:lockingexample}
\end{figure}


\section{Request priorities}
\label{subsec:buffer}
To avoid starvation, loops and deadlocks, each node maintains a local buffer $\mathcal{B}$, containing a queue of rotation requests received from its neighborhood. Given the concurrency locks, described in the previous section, it is easy to see that a node can participate in a rotation originated by: itself, its right child, its left child, one of its four grandchildren or eight grand-grandchildren. Thus, $|\mathcal{B}| = 15$. In addition, every time a nodes $u$ rotates, the buffer of the nodes whose local variables (links) were changed due to this rotation must be updated, i.e., buffers contents must be exchanges among nodes with link changes.

Each buffer entry is an array consisting of: 
\begin{enumerate}
 \item Super-round and round sequence numbers;
 \item ID of the node that requested the rotation (node \textbf{level-one}); 
 \item ID of the node through which the request arrived (node \textbf{level-two}); 
 \item ID of the farthest node participating in the rotation (node \textbf{level-three}); 
 \item ID of the destination node (or source node) of the splay.
\end{enumerate}

\begin{defn}\label{def:priority}
\textbf{Priority routine}: Consider a node $u$ and its buffer $\mathcal{B}_u$ and the set of rotation requests received by $u$ before some time-slot $\tau \in \mathcal{R}$. The possible request sources are: $u$ (itself), $2 \times c(u)$ (two children), $4 \times g(u)$ (four grandchildren), and $8 \times gg(u)$ (eight grandchildren). Each request $\beta(s,t_s)$, where $s$ is the source and $t$ is the round sequence number (note that each rotation request may belong to a different round), is acknowledged by $u$ using the following rule, to which we refer as the priority routine:
\begin{enumerate}
\item \textit{Non-decreasing super-round sequence number};
\item \textit{Non-decreasing round sequence number};
\item \textit{Hierarchical priority}: within the same round and super-round, priority is given in the following order: $u$ (itself), $c(u)$ (children), $g(u)$ (grandchildren), and $gg(u)$ (grandchildren);  
\item \textit{Ascending ID}: within the same round and hierarchy level, priority is given to nodes with the smallest ID.
\end{enumerate}
\end{defn}

The priority routine is used to avoid starvation, loops, and deadlocks (See Lemmas \ref{claim:noLoops}, \ref{claim:noDeadlocks} and \ref{claim:round}). Figure \ref{fig:bufferexample} shows an example of buffer contents, sorted using the priority routine, when several nodes request a rotation. Nodes in gray are the ones waiting to rotate, and all requests belong to the same round and super-round, so their sequence numbers are omitted.

\begin{figure}[H]
    \centering
    \includegraphics[width=0.6\linewidth]{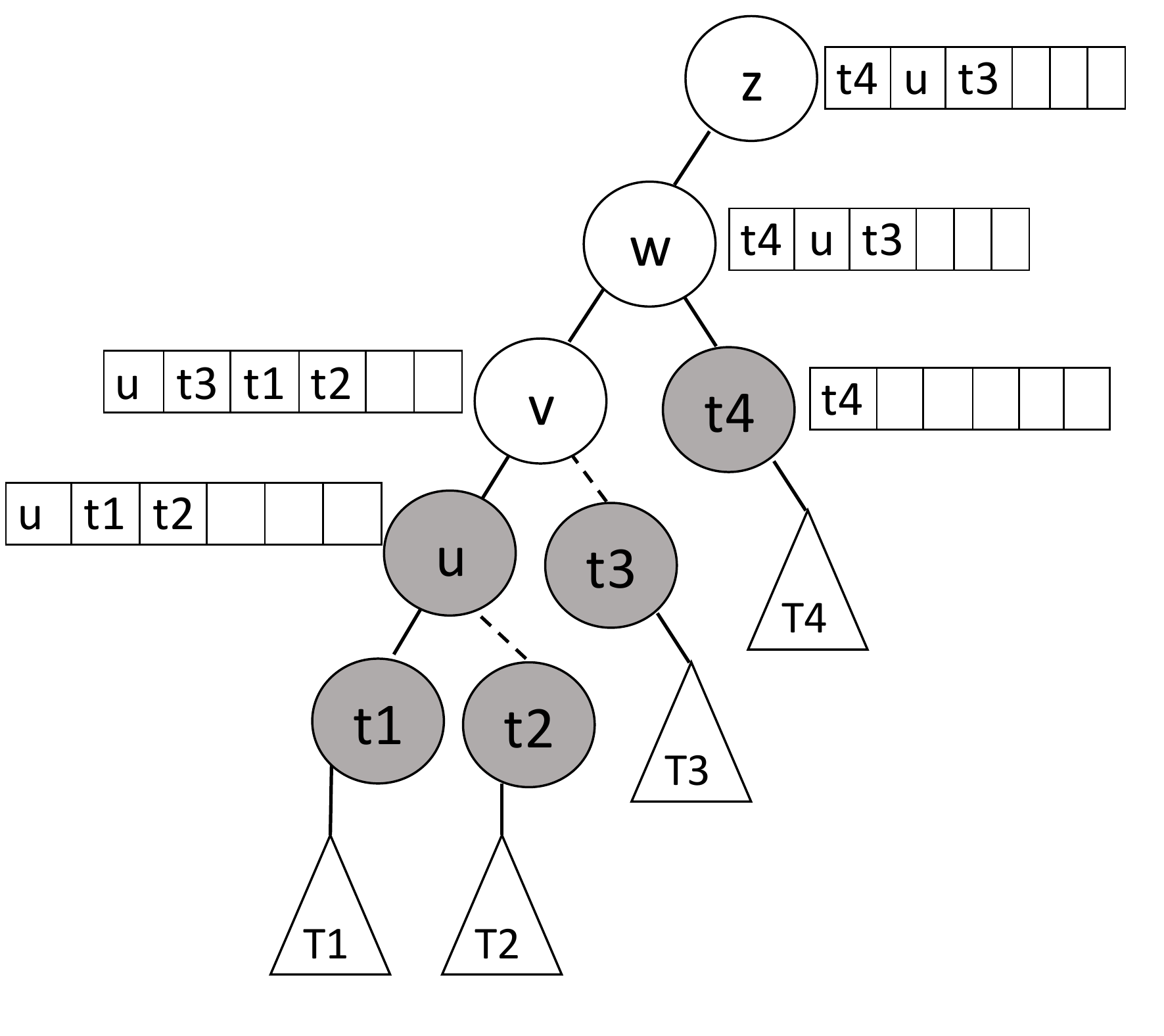}
    \caption{ Local buffers sorted by the \textit{priority routine} (nodes in gray requested a rotation).}
    \label{fig:bufferexample}
\end{figure}

The priority routine is used to avoid starvation, loops, and deadlocks (See Lemmas \ref{claim:noLoops}, \ref{claim:noDeadlocks} and \ref{claim:round}). Figure \ref{fig:bufferexample} shows an example of buffer contents, sorted using the priority routine, when several nodes request a rotation. Nodes in gray are the ones waiting to rotate, and all requests belong to the same round and super-round, so their sequence numbers are omitted.

Consider a time-slot $\tau$ and a node $u$. The buffer of $u$ is \textit{updated} when any of the following events occur: (1) Upon receiving or generating a splay request; (2) Upon receiving a rotation request or a lock request; (3) After generating a rotation request: when $u$ wants to rotate, it must place its own request in the buffer; 
(4) After completing a rotation $\beta(u)$, it is removed from all buffers; (5) Upon changing a local variable (link): If there is a new parent, any entry in which the new parent is level-one must be added to the buffer. If there is an entry from the old parent, it must be removed. If there is a new child, the new child must have forwarded to the new parent (for example, in the rotation ack message) the list of requests in its buffer for which it is level-one, level-two or level-three. Any entry containing the previous child must be removed and the requests from the new child's buffer must be inserted.

\section{Distributed routines}
\label{subsec:distributed}

Consider a splay request $\mathcal{S}(s,d)\in \mathcal{F}$.
Algorithm \ref{alg:splay} presents the distributed routines for a source node $s$ to perform a splay with the destination node $d$. Essentially, $s$ sends a splay request to $d$ and both nodes keep requesting a rotation to its parent, until $d(s,d)=1$.

Algorithms \ref{alg:rotation-u}, \ref{alg:rotation-v}, \ref{alg:rotation-w}, \ref{alg:rotation-wp} and\ref{alg:rotation-x} present the communicating routines for a rotation request. Consider a rotation $\beta(u)$, generated at time-slot $\tau_1$, without any other concurrent rotations in the network. The following sequence of steps takes place (consider notation from zig-zig example in Figure~\ref{fig:rotations}, with $z=w.p$, $t_i=$ root of sub-tree $T_i$, $\mathcal{B}_x=$ buffer of node $x$): 
\begin{enumerate}
\singlespacing
    \item[$\tau_1$:] $u$ sends $\beta$-request to $v$;
    \item[$\tau_2$:] $v$ forwards $\beta$-request to $w$; 
    \item[$\tau_3$:] $w$ sends lock-request to $z$;
    \item[$\tau_4$:] $z$ sends lock-ack to $w$, locks itself and updates link $w.l\lor w.r = u$;
    \item[$\tau_5$:] $w$ sends lock-ack to $v$, locks itself, updates links $w.p=v \land w.l=t_3$ and sends link-change to $t_3$;
    \item[$\tau_6$:] $v$ sends lock-ack to $u$, locks itself, updates links $v.p=u \land v.l=t_2 \land v.r=w$ and sends link-change to $t_2$;
    \item[$\tau_6$:] $t_3$ updates link $t_3.p=w$ and buffer $\mathcal{B}_{t_3} \cup \mathcal{B}_{w} \setminus \mathcal{B}_v$, and sends buffer-change to $w$;
    \item[$\tau_7$:] $u$ locks itself, updates links $u.p=w.p \land u.r=v$ and buffer $\mathcal{B}_u \cup \mathcal{B}_{v} \setminus \mathcal{B}_{t_2}$;
    \item[$\tau_7$:] $u$ frees itself and sends $\beta$-ack to $v$ and buffer-change to $z$;
    \item[$\tau_7$:] $t_2$ updates link $t_2.p=v$ and buffer $\mathcal{B}_{t_2} \cup \mathcal{B}_{v} \setminus \mathcal{B}_u$, and sends buffer-change to $v$;    
    \item[$\tau_8$:] $z$ frees itself and updates its buffer buffer $\mathcal{B}_z \cup \mathcal{B}_{u} \setminus \mathcal{B}_w$;
    \item[$\tau_8$:] $v$ frees itself; forwards $\beta$-ack to $w$ and updates buffer $\mathcal{B}_v \cup \mathcal{B}_{t_2} \cup \mathcal{B}_{w} \setminus \{t_3\} \setminus \{u\}$;    
    \item[$\tau_9$:]\label{item:lastStep} $w$ frees itself and updates buffer $\mathcal{B}_w \cup \mathcal{B}_{t_3} \setminus \mathcal{B}_{v}$;

\end{enumerate}

This sequence of steps takes at least 9 time-slots, but the steps are not necessarily consecutive and ordered in this way, due to concurrency. In Lemma \ref{claim:round}, we provide an upper bound on the number of time-slots a rotation takes, considering concurrency.

    \begin{algorithm}[H]
    \begin{algorithmic}
    \scriptsize
    \State\textbf{s:} request-splay(d)
    \While{$s\neq \alpha(s,d)$ and $s.p\neq d$}
        \State\textbf{s:} \textbf{generate} $\beta$(s)
    \EndWhile
    \If{$s.p = d$}
        \State wait message
    \Else
        \If{$s$ is no longer $\alpha(s,d)$}
            \State go back to \textbf{while}
        \Else
            \If{$s.r = d $ or $s.l = d$}
                \State send message
            \EndIf  
        \EndIf
    \EndIf
    \While{$d\neq \alpha(s,d)$ and $d.p\neq s$}
        \State\textbf{d:} \textbf{generate} $\beta$(d)
    \EndWhile
    \If{$d.p = s$}
        \State wait message
    \Else
        \If{$d$ is no longer $\alpha(s,d)$}
            \State go back to \textbf{while}
        \Else
            \If{$d.r = s $ or $d.l = s$}
                \State send message
            \EndIf  
        \EndIf
    \EndIf
    \end{algorithmic}
    \caption{Splay $\mathcal{S}(s,d)$}\label{alg:splay}
    \end{algorithm}

    \begin{algorithm}[H]
    \begin{algorithmic}
    \scriptsize
    \State Upon generating ($\beta(u)$)
    \State \textbf{insert-buffer}($\beta(u)$)
    \State \textbf{request-$\beta$}($\beta(u)$) to \textbf{v}
    \If{$\mathcal{B}[0]=\beta(u)$ $\land$ ack-$\beta$($\beta(u)$) from \textbf{v}}
    \State \textbf{rotate}($\beta(u)$)
    \EndIf    
    \end{algorithmic}
    \caption{Rotation $\beta(u)$ at $u$}\label{alg:rotation-u}
    \end{algorithm}    
    
    \begin{algorithm}[H]
    \begin{algorithmic}
    \scriptsize
    \State Upon receiving $\beta(u)$
    \State \textbf{insert-buffer}($\beta(u)$)
    \If{$v \neq \alpha(u,d_u)$}
        \State forward \textbf{request-$\beta$($\beta(u)$)} to \textbf{w}
        \If{$\mathcal{B}[0]=\beta(u)$ $\land$ ack-$\beta$($\beta(u)$) from \textbf{w}}
            \State \textbf{ack-$\beta$}($\beta(u)$) to \textbf{u}
            \State \textbf{rotate}($\beta(u)$)
        \EndIf
    \Else
        \State \textbf{request-lock}($\beta(u)$) to \textbf{w}
        \If{ $\mathcal{B}[0]=\beta(u)$ $\land$ ack-lock($\beta(u)$) from \textbf{w}}
            \State \textbf{ack-$\beta$}($\beta(u)$) to \textbf{u}
            \State \textbf{rotate}($\beta(u)$)
        \EndIf
    \EndIf
    
    \singlespacing
    \State Upon connecting to a new child
    \State {\tabs}send \textbf{link-chage} to new child
    
    \singlespacing
    \State Upon receiving \textbf{buffer-change}(from new child)
    \State {\tabs}update buffer
    \end{algorithmic}
    \caption{Rotation $\beta(u)$ at $v$}\label{alg:rotation-v}
    \end{algorithm}

    \begin{algorithm}[H]
    \begin{algorithmic}
    \scriptsize
    \State Upon receiving $\beta(u)$
    \State\textbf{insert-buffer}($\beta(u)$)
    \If{$v \neq \alpha(u,d_u)$}
        \State\textbf{request-lock} to \textbf{w.p}
        \If{$\mathcal{B}[0]=\beta(u)$ $\land$ ack-lock($\beta(u)$) from \textbf{w.p}}
            \State \textbf{ack-$\beta$}($\beta(u)$) to \textbf{v}
            \State \textbf{rotate}($\beta(u)$)
        \EndIf
    \Else
        \If{$\mathcal{B}[0]=\beta(u)$}
            \State \textbf{ack-lock}($\beta(u)$) to \textbf{v}
        \EndIf
    \EndIf

    \singlespacing
    \State Upon connecting to a new child
    \State {\tabs}send \textbf{link-chage} to new child
    
    \singlespacing
    \State Upon receiving \textbf{buffer-change}(from new child)
    \State {\tabs}update buffer
    \end{algorithmic}
    \caption{Rotation $\beta(u)$ at $w$}\label{alg:rotation-w}
    \end{algorithm}

   \begin{algorithm}[H]
    \begin{algorithmic}  
    \scriptsize
    \State Upon receiving \textbf{request-lock}($\beta(u)$)
    \State \textbf{insert-buffer}($\beta(u)$)
    \If{$\mathcal{B}[0]=\beta(u)$}
        \State \textbf{ack-lock($\beta(u)$)} to \textbf{w}
    \EndIf
    \singlespacing
    \State Upon receiving \textbf{buffer-change}(from u)
    \State {\tabs}update buffer
    \end{algorithmic}
    \caption{Rotation $\beta(u)$ at $w.p$}\label{alg:rotation-wp}
    \end{algorithm}
    
    \begin{algorithm}[H]
    \begin{algorithmic}  
    \scriptsize
    \State Upon receiving \textbf{link-change}(from $x$ about $relationship$)
    \State {\tabs} y.relationship = $x$
    \State {\tabs} send \textbf{buffer-update} to $x$
    \end{algorithmic}
    \caption{Link change at $y$}\label{alg:rotation-x}
    \end{algorithm}
\chapter{Concurrency Analysis:}\label{chap:analysis}
The analysis of the distributed and concurrent SplayNet is structured as follows. In Section \ref{subsec:loopsDeadlocks}, we show that the \textit{priority routine} prevents loops (Lemma \ref{claim:noLoops}) and deadlocks (Lemma \ref{claim:noDeadlocks}) from occurring between concurrent rotations. In Section \ref{subsec:round}, we show that the duration of a round, i.e., the time between a node requesting and completing a rotation, is $O(\log{m})$, where $m$ is the number of concurrent splay requests in the network (Lemma \ref{claim:round}). In Section \ref{subsec:amortized}, we compute the total amortized average cost of a splay request (Theorem \ref{thm:totalcost}) and as a function of the \textit{empirical entropies} of source and destination nodes of the splay requests (Theorem \ref{thm:finalEntropies}).

\section{Loops and deadlocks}
\label{subsec:loopsDeadlocks}

\begin{defn}\label{def:loop} \textbf{Infinite loop:} Consider a SplayNet $\mathcal{T}$, a set of concurrent splay requests $\mathcal{F}$, a splay request $\mathcal{S}(a_i,d_{a_i}) \in \mathcal{F}$, and a time-slot $\tau_1$, such that distance $d_{\tau_1}(a_i,d_{a_i}) > 1$. An infinite loop is said to occur in $\mathcal{T}$ when there are an infinite number of (possibly non-consecutive) time-slots $\tau_k > \tau_1$, in which the following two conditions hold:
\begin{enumerate}
    \item The parent or grandparent of $a_i$ remains the same in every (non-consecutive) time-slot $\tau_k$, i.e., $a_i.p(\tau_k)=a_i.p(\tau_1) \lor a_i.p.p(\tau_k)=a_i.p.p(\tau_1), \forall \tau_k$. 
    \item The distance to the destination does not decrease relative to time-slot $\tau_1$, i.e. $d_{\tau_k}(a_i,d_{a_i}) \geq d_{\tau_1}(a_i,d_{a_i}), \forall \tau_k$. 
\end{enumerate}
\end{defn}

\begin{lemma}\label{claim:noLoops}
Priority routine prevents loops from occurring in the network. 
\end{lemma}
\begin{proof}
Consider a SplayNet $\mathcal{T}$ and a set of concurrent splay requests $\mathcal{F}$ in super-round $\mathcal{R}$.
We divide the loop possibilities in two main cases: (Case 1) a loop between two nodes; and (Case 2) a loop between more than two nodes. 

\begin{enumerate}
    \item[Case 1:] \label{case:loop2nodes} \textbf{Loop between two nodes:}
    Since $\mathcal{T}$ has no cycles in a given time-slot, this scenario can only occur between two nodes that participate or are locked in each others rotation requests. There can be three configurations:

    \begin{enumerate}
        \item \textbf{Loop with a parent:} \label{case:loopParent} $\mathcal{S}(b,d_{b})$, $\mathcal{S}(a,d_{a})\in\mathcal{F}$ $\mid b = a.p(t)$ in round $t \in \mathcal{R}$. Consider the following rotations requests: $\beta_{t}(b)$, $\beta_{t}(a)$, $\beta_{t+1}(b)$, $\beta_{t+1}(a)$. Note that \textit{priority routine} puts $\beta_{t}(b)$ before $\beta_t(a)$. There are four possibilities:
        
        \begin{enumerate}
            \item $\beta_t(a)$ is a zig: If $\beta_t(a)$ is a zig, then it is the last rotation sourced at $a$ in splay $\mathcal{S}(a,d_{a})$, since $b = \alpha_t(a,d_a)$. So, $a$ and $b$ cannot be in an infinite loop.
            
            \item $\beta_t(b)$ is a zig: If $\beta_t(b)$ is a zig, then it is the last rotation sourced at $b$ in splay $\mathcal{S}(b,d_{b})$, since $b.p(t) = \alpha_t(b,d_b)$ or $b.p = s_b$ for some $\mathcal{S}(s_b,b)\in\mathcal{F}$. So, $a$ and $b$ are not in an infinite loop.
            
            \item $\beta_t(b)$ is a zig-zig: If $\beta_t(b)$ is a zig-zig (Figure \ref{fig:loopparentzigzig}), $a$ is a \textit{carried child} and moves one or two hops upwards with $\beta_t(b)$. If $\beta_t(a)$ is performed right after $\beta_t(b)$, then $\beta_t(a)$ is a zig-zag. After $\beta_t(a)$, $d(a,d_a)$ decreases by two, and $d(b,d_b)$ does not change or increases by one. In both cases, $a = b.p(t+1)$, and $\beta_{t+1}(a)$ has priority over $\beta_{t+1}(b)$. After $\beta_{t+1}(a)$, $b$ is carried one or two hops upwards. In this way, even though $a$ and $b$ keep rotating with each other through the rounds, the distances to their destinations decrease. Therefore, by Definition~\ref{def:loop}, $a$ and $b$ are not in an infinite loop.
                
            \item $\beta_t(b)$ is a zig-zag: If $\beta_t(b)$ is a zig-zag (Figure \ref{fig:loopparentzigzag}), $a$ will be carried one hop upwards. If $\beta_t(a)$ is performed right after $\beta_t(b)$, then $\beta_t(a)$ is a zig-zag. After $\beta_t(a)$, $d(a,d_a)$ decreases by two, and $d(b,d_b)$ increases by one. In both cases, $a = b.p(t+1)$, and $\beta_{t+1}(a)$ has priority over $\beta_{t+1}(b)$. After $\beta_{t+1}(a)$, $b$ is carried one or two hops upwards. In this way, even though $a$ and $b$ keep rotating with each other through the rounds, the distances to their destinations decrease. Therefore, $a$ and $b$ are not in an infinite loop.
        \end{enumerate}
        
        \item \textbf{Loop with a grandparent:} \label{case:loopGP} $\mathcal{S}(c,d_{c})$, $\mathcal{S}(a,d_{a})\in\mathcal{F}$ $\mid b = a.p(t), c=b.p(t)$ in round $t \in \mathcal{R}$. Consider the following rotations requests: $\beta_{t}(c)$, $\beta_{t}(a)$, $\beta_{t+1}(c)$, $\beta_{t+1}(a)$. Note that \textit{priority routine} puts $\beta_{t}(c)$ before $\beta_t(a)$. There are two possibilities\footnote{We do not consider the case where $\beta(a)$ as a zig because it's clear that $c$ does not participate in $\beta(b)$ in this case.}:
        \begin{enumerate}
            \item $\beta(c)$ is a zig-zig: After $\beta_t(c)$, there are two possibilities:
            \begin{enumerate}
            
            \item $b$ is an \textit{abandoned child}: (Figure \ref{fig:loopzigzig1}) $c$ will not participate in $\beta_t(a)$, because $d_t(c,a) = 3$. So, there is no loop.
            
            \item $b$ is a \textit{carried child}:       
               \\
                    I.{  } $\beta_t(a)$ is a zig-zig (Figure \ref{fig:loopzigzig2.1}): After $\beta_t(a)$, $d(c,d_c)$ will increase by two. In the next round, $\beta_{t+1}(a)$ has priority over $\beta_{t+1}(c)$, because $a = c.p.p(t+1)$. Thus, $c$ will be carried by $\beta_{t+1}(a)$ and then will perform $\beta_{t+1}(c)$. After those four rotations, $d(c,d_c)$ decreases by at least 3. Thus, the distance to the destination decreases for nodes $a$ and $c$. Therefore, $a$ and $c$ are not in an infinite loop.
                    \\
                    II.{  } $\beta_t(a)$ is a zig-zag (Figure \ref{fig:loopzigzig2.2}): After $\beta_t(c)$, $d_{t+1}(c,a) = 2$ and $d_{t}(a,d_a)$ is decreased by two. After $\beta_{t+1}(a)$, $a = c.p(t+1)$. Thus, any further rotations fit into Case 1(a). 
            \end{enumerate}
            
            \item $\beta_t(c)$ is a zig-zag: (Figure \ref{fig:loopzigzag}) $d_{t_1}(c,a)=3$. Thus, $c$ will not participate in $\beta_t(a)$, regardless of $\beta_t(a)$ being a zig-zig or zig-zag.
        \end{enumerate}
        
        \item \textbf{Loop with a sibling:}  $\mathcal{S}(t3,d_{t3})$, $\mathcal{S}(a,d_{a})\in\mathcal{F}$ $\mid b = a.p = t3.p(t)$ in round $t \in \mathcal{R}$. W.l.g., $a = b.l(t)$ and $t3 = b.r(t)$. Consider the following rotations requests: $\beta_{t}(a)$, $\beta_{t}(t3)$, $\beta_{t+1}(a)$, $\beta_{t+1}(t3)$. Note that \textit{priority routine} puts $\beta_{t}(a)$ before $\beta_t(t3)$.
            \begin{enumerate}
                \item $\beta_t(a)$ or $\beta_t(t3)$ is a zig: one of source nodes will terminate its rotations, and no loop will be generated (Figure \ref{fig:loopsiblings}).
                \item $\beta_t(a)$ is a zig-zig: after $\beta_t(a)$, $d_{t+1}(a,t3) = 3$, thus $a$ will not participate in $\beta_t(t3)$. Therefore, $a$ and $t3$ are not in an infinite loop.
                \item $\beta_t(a)$ is a zig-zag: After $\beta_t(a)$, $a = t3.p.p(t+1)$. Therefore, $\beta_t(t3)$ fits into Case 1(b).
            \end{enumerate}
    \end{enumerate}    
    
\item[Case 2:] \textbf{Loop among $\geq 3$ nodes:} Let's assume the existence of a loop comprised of $k$ nodes: $\{a_1,a_2,\ldots, a_k\} \mid 2 < k \leq n$. To be in an infinite loop, $\exists$ infinite time-slots $\{\tau_k\}$, such that no $a_i$ advances towards its destination and nodes become each others parent/grandparent in a circular list: $a_i = a_{i+1}.p(\tau_k)$ or $a_i = a_{i+1}.p.p(\tau_k)$, $\ldots$, and $a_k = a_{1}.p(\tau_k')$ or $a_k = a_{1}.p.p(\tau_k')$. Consider node $a_k$ in time slot $\tau'$. For it to become a parent or grandparent of $a_i$, it must exchange places with $\{a_{k-1},\ldots,a_2\}$ because all these nodes are ancestors of $a_k$ and descendants of $a_1$ and, therefore, are between $a_k$ and $a_1$. This means that, $a_k$ will exchange parent/child or grandparent/grandchild roles and, therefore, be in a loop with each $a_{i}$. This is a contradiction, because it was proven in Case 1 that there can be no loop between any pair of nodes in $\mathcal{T}$.
\end{enumerate}

\begin{figure}
	\centering
	\begin{subfigure}{0.4\textwidth} 
		\includegraphics[width=\textwidth]{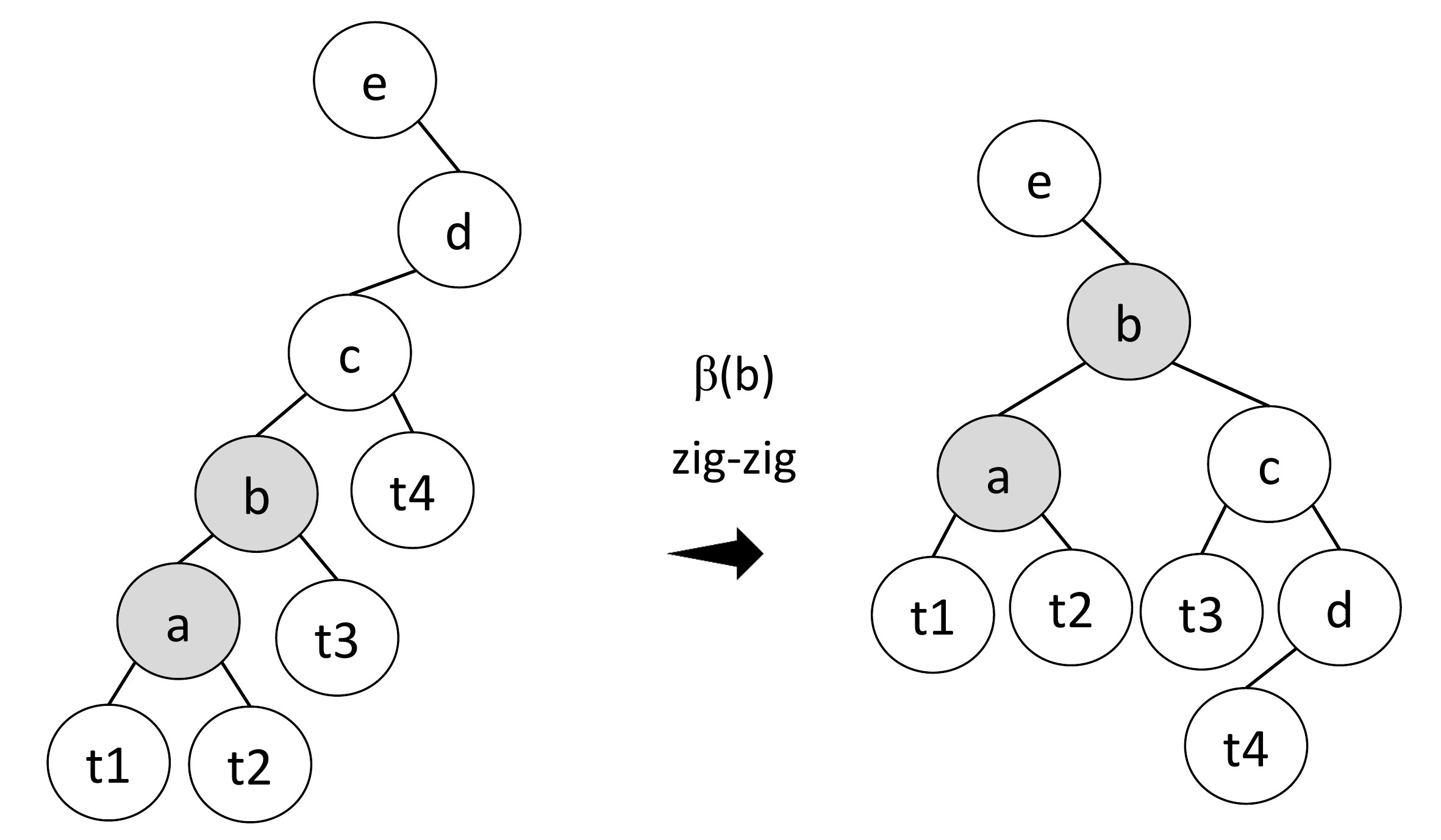}\caption{}
		\label{fig:bzigzig1}
	\end{subfigure}
	\vspace{1em} 
	\begin{subfigure}{0.4\textwidth} 
		\includegraphics[width=\textwidth]{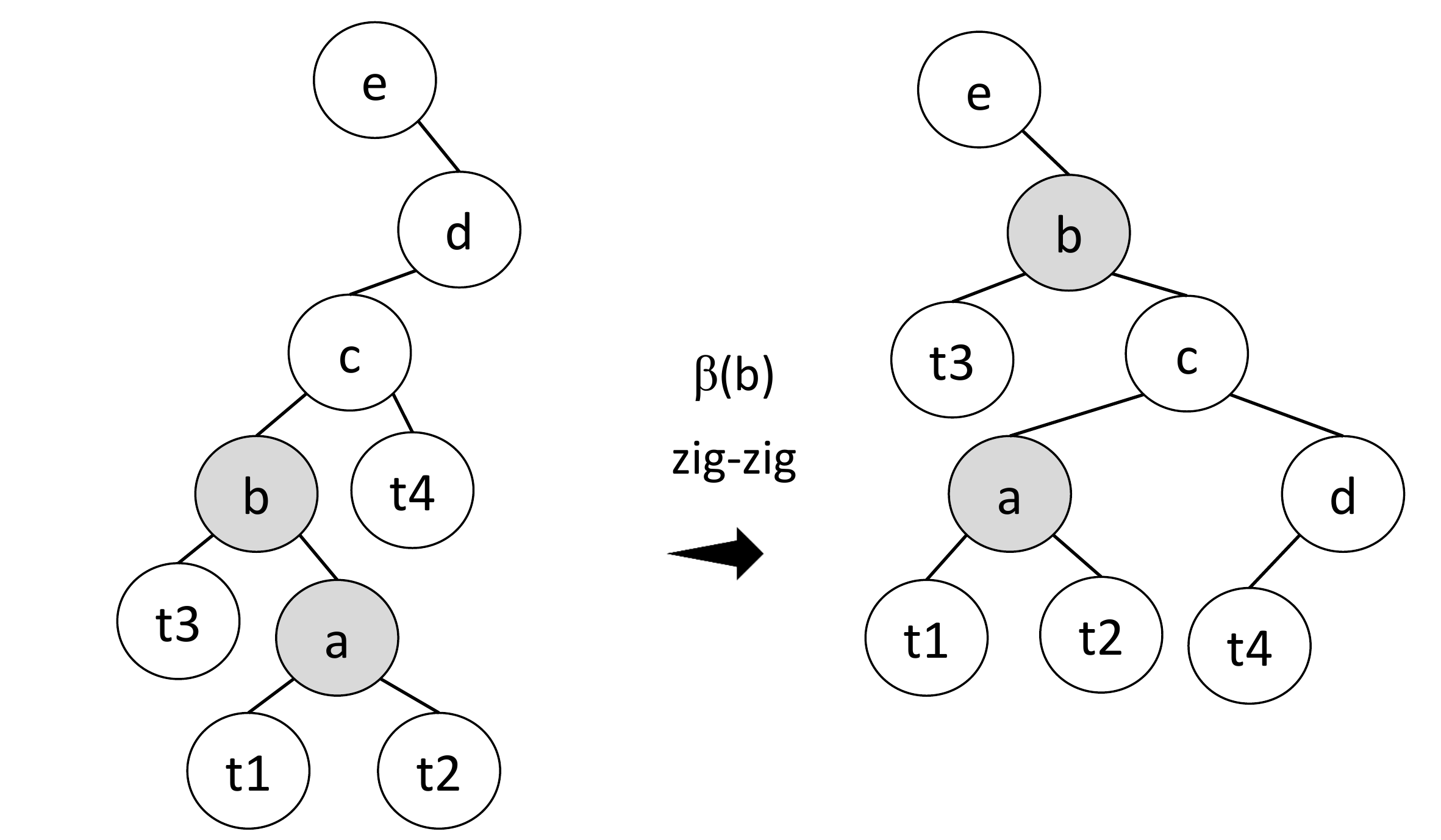}\caption{}
		\label{fig:bzigzig2}
	\end{subfigure}	
	\begin{subfigure}{0.4\textwidth} 
		\includegraphics[width=\textwidth]{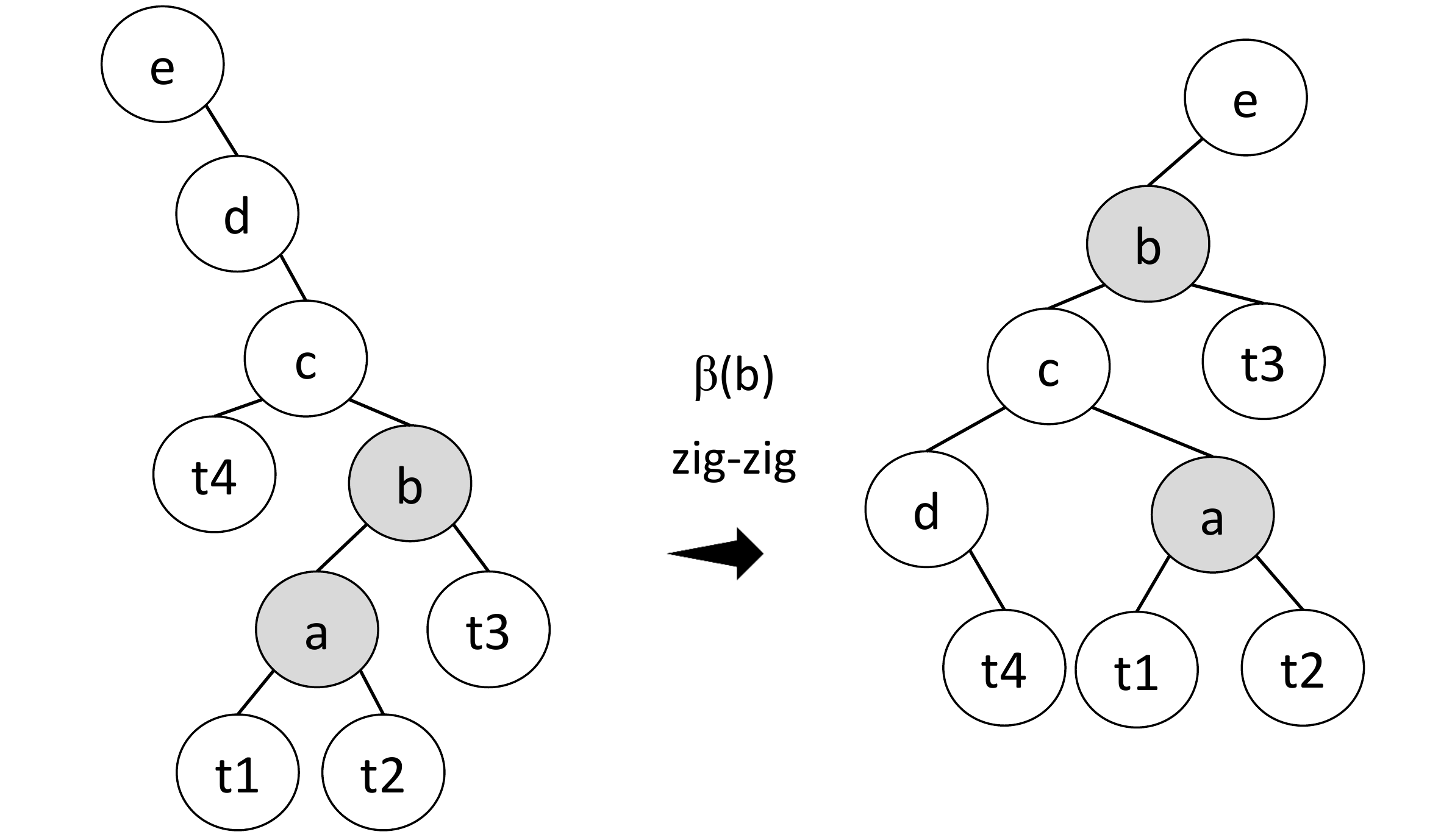}\caption{}
		\label{fig:bzigzig3}
	\end{subfigure}	
	\begin{subfigure}{0.4\textwidth} 
		\includegraphics[width=\textwidth]{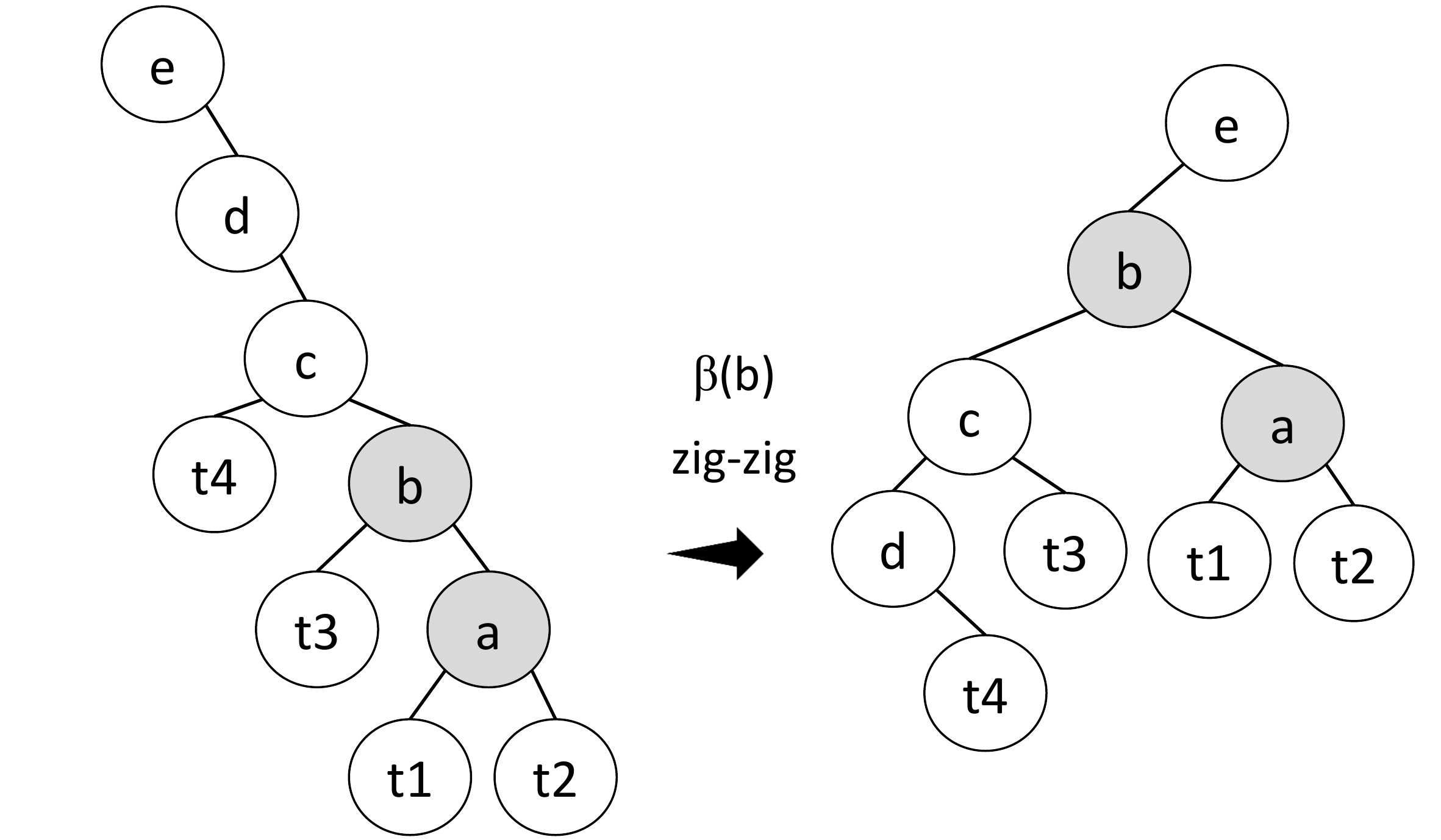}\caption{}
		\label{fig:bzigzig4}
	\end{subfigure}
	\caption{Loop with a parent (Case 1.a.iii): $b = a.p(t), \beta_t(b)$ is a zig-zig. Even though $a,b$ rotate with each other repeatedly, both move upwards.} 
	 \label{fig:loopparentzigzig}
\end{figure}

\begin{figure}
	\centering
	\begin{subfigure}{0.4\textwidth} 
		\includegraphics[width=\textwidth]{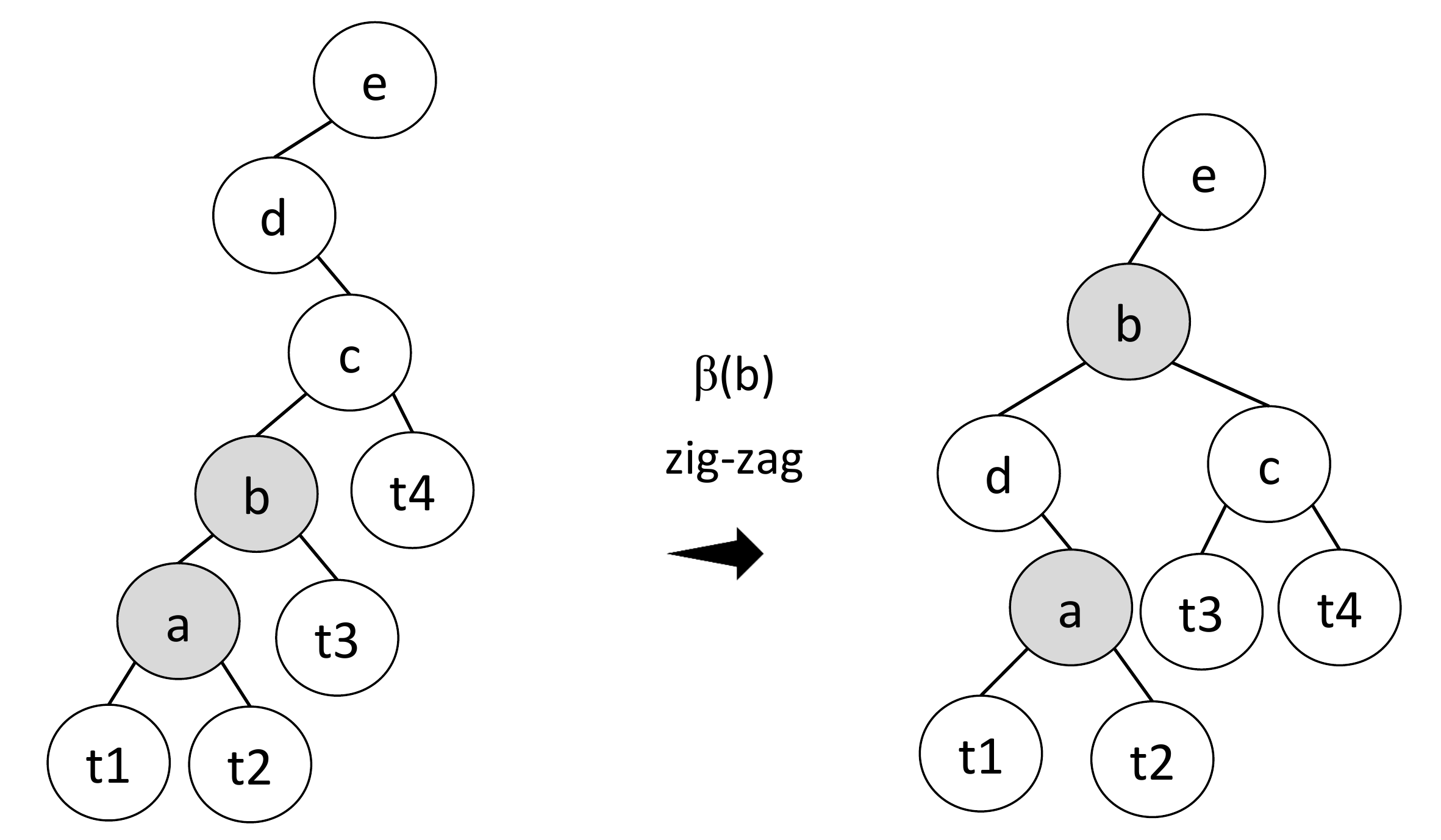}\caption{}
		\label{fig:bzigzag1}
	\end{subfigure}
	\vspace{1em} 
	\begin{subfigure}{0.4\textwidth} 
		\includegraphics[width=\textwidth]{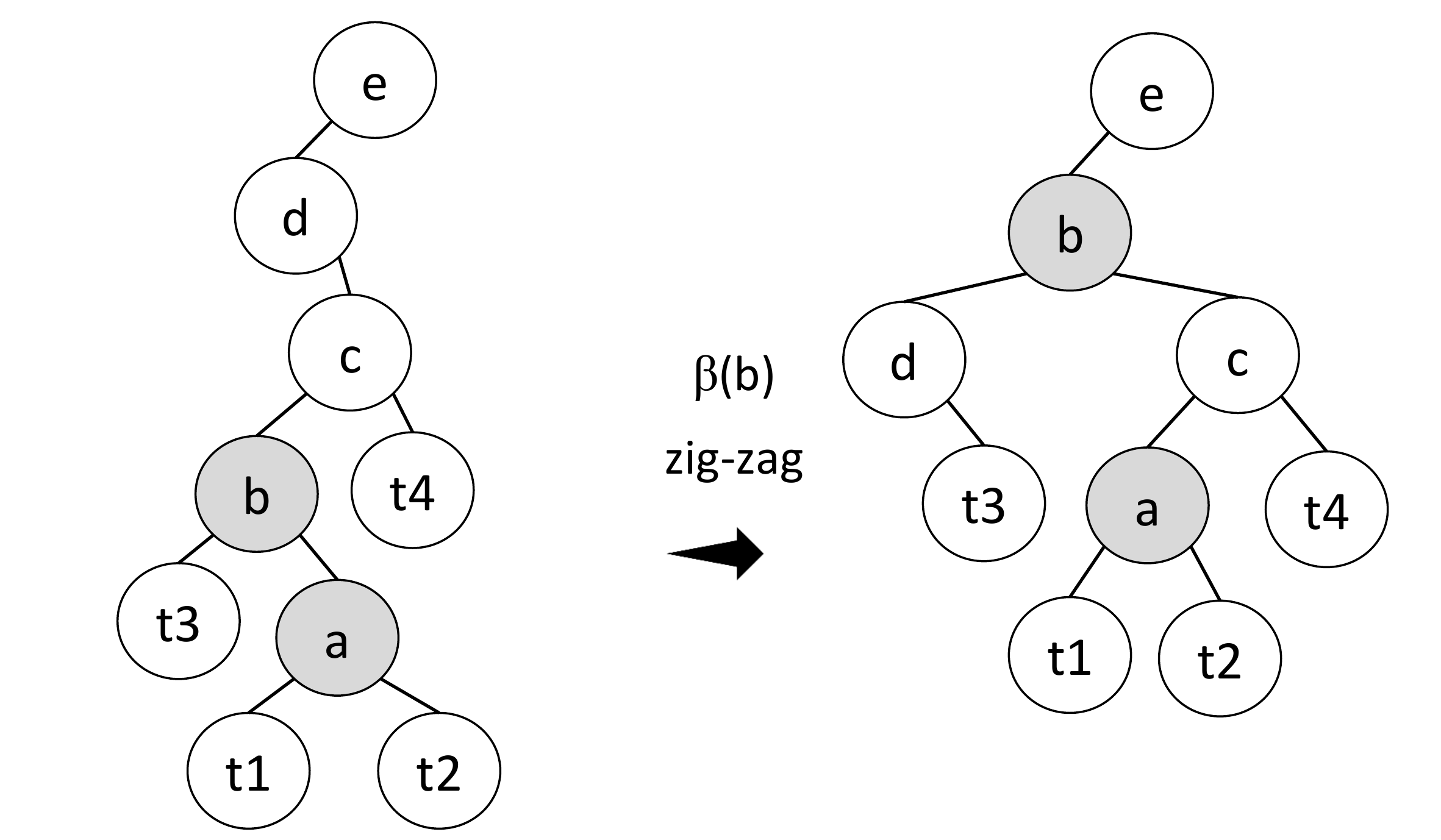}\caption{}
		\label{fig:bzigzag2}
	\end{subfigure}	
	\begin{subfigure}{0.4\textwidth} 
		\includegraphics[width=\textwidth]{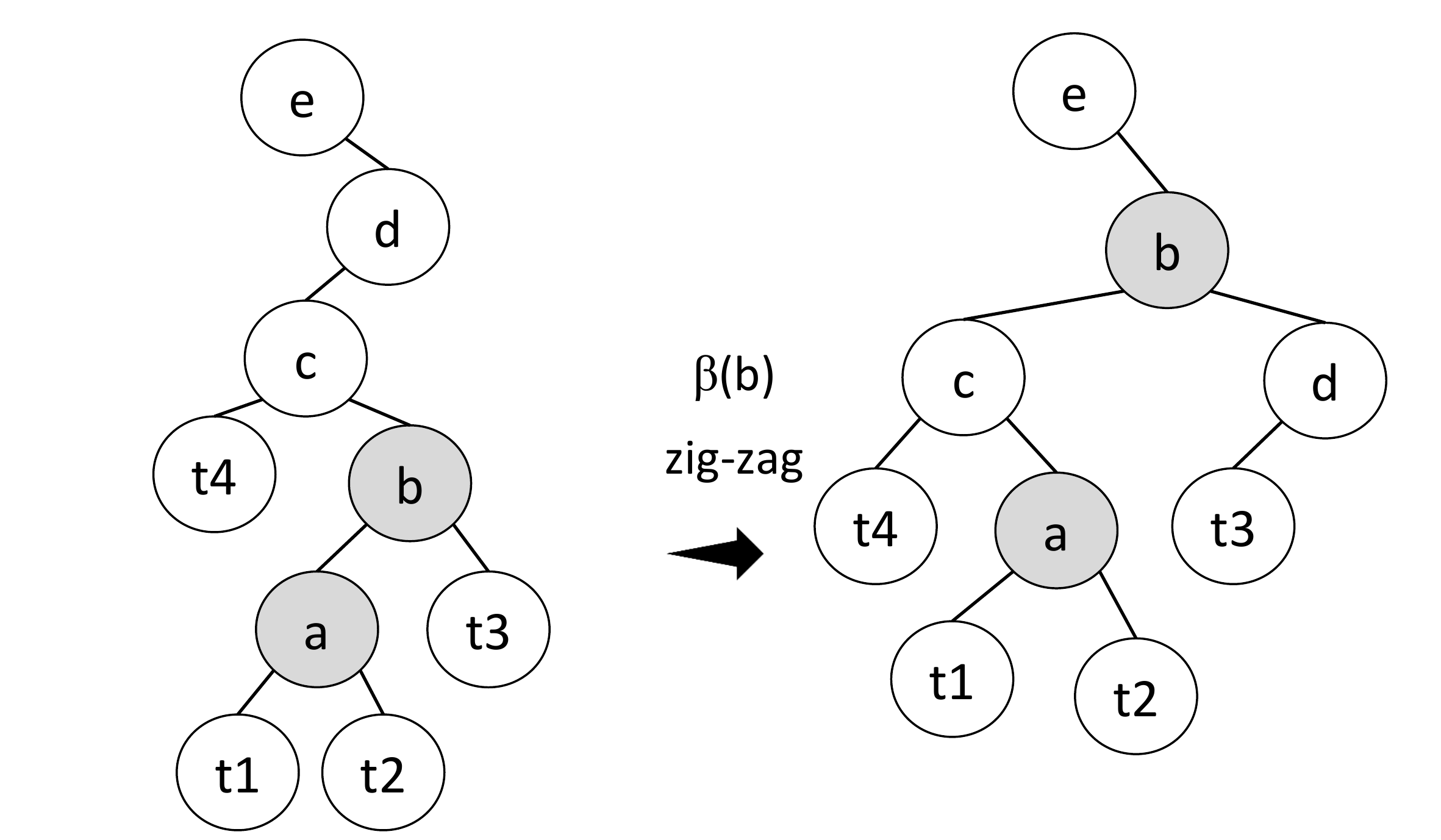}\caption{}
		\label{fig:bzigzag3}
	\end{subfigure}	
	\begin{subfigure}{0.4\textwidth} 
		\includegraphics[width=\textwidth]{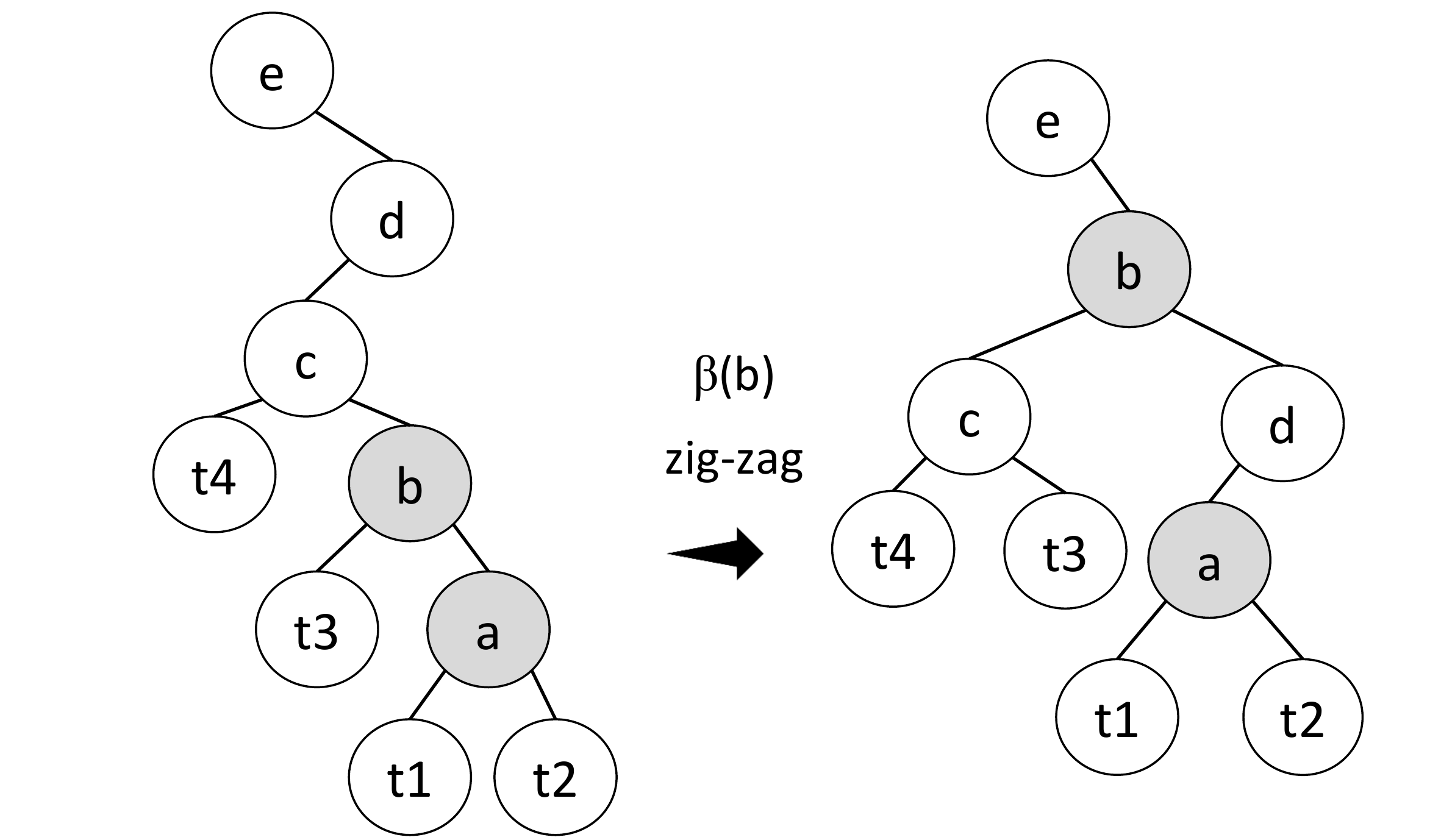}\caption{}
		\label{fig:bzigzag4}
	\end{subfigure}
    \caption{Loop with a parent (Case 1.a.iv): $b=a.p(t), \beta_t(b)$ is a zig-zag. Even though $a,b$ rotate with each other multiple times, both advance towards their destinations.}
    \label{fig:loopparentzigzag}
\end{figure}

\begin{figure}
	\centering
	\begin{subfigure}{0.4\textwidth} 
		\includegraphics[width=\textwidth]{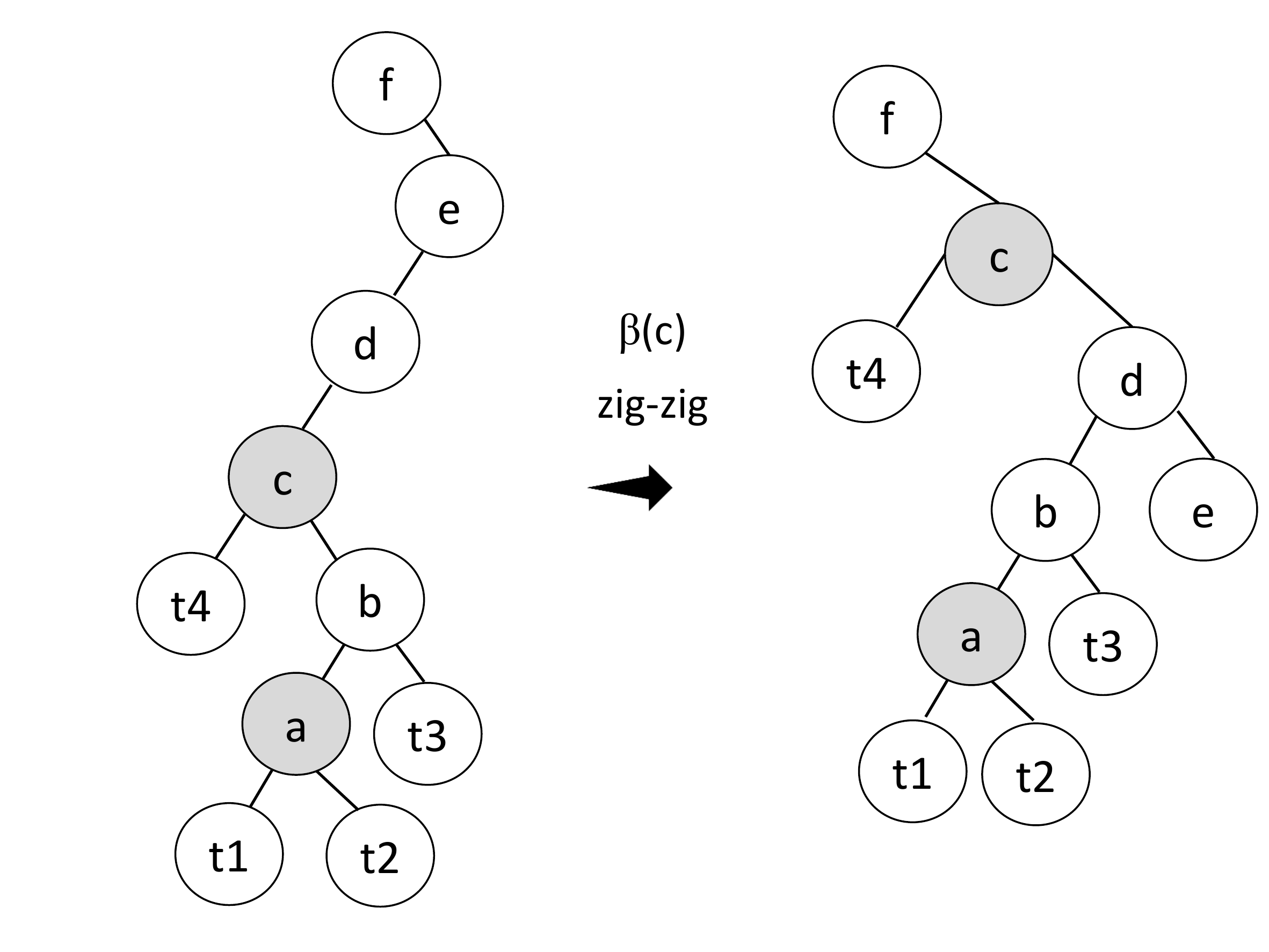}\caption{}
		\label{fig:zigzig2}
	\end{subfigure}
	\vspace{1em} 
	\begin{subfigure}{0.4\textwidth} 
		\includegraphics[width=\textwidth]{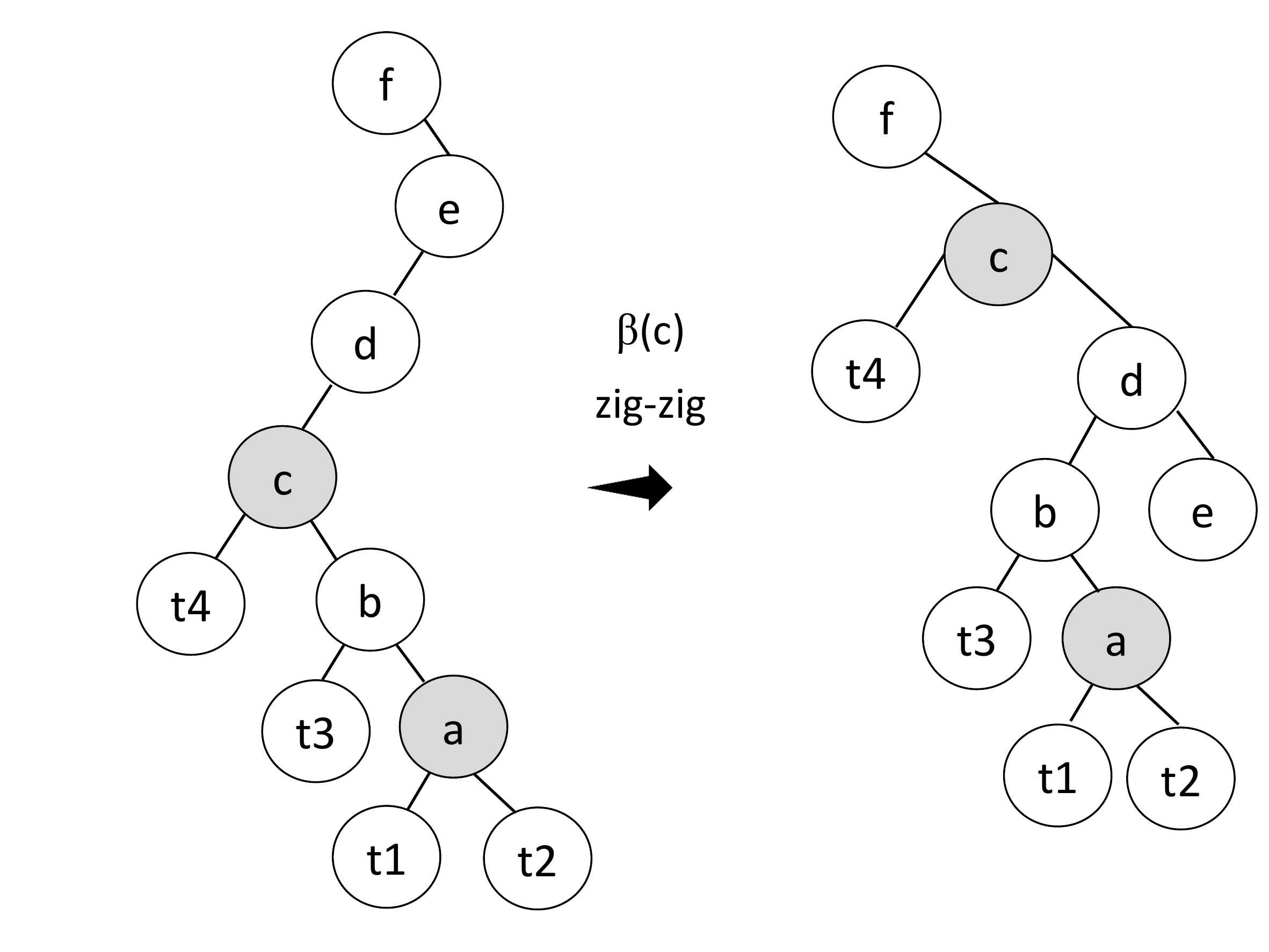}\caption{}
		\label{fig:zigzig3}
	\end{subfigure}	
   \caption{Loop with a grandparent (Case 1.b.i.A): $c=b.p=a.p.p(t), \beta_t(c)$ is a zig-zig, and $b$ is an \textit{abandoned child}. In both cases, $d_{t+1}(c,a)=3$ after the rotation.}
    \label{fig:loopzigzig1}
\end{figure}

 \begin{figure}[!ht]
    \begin{center}
    {\includegraphics[width=1\columnwidth]{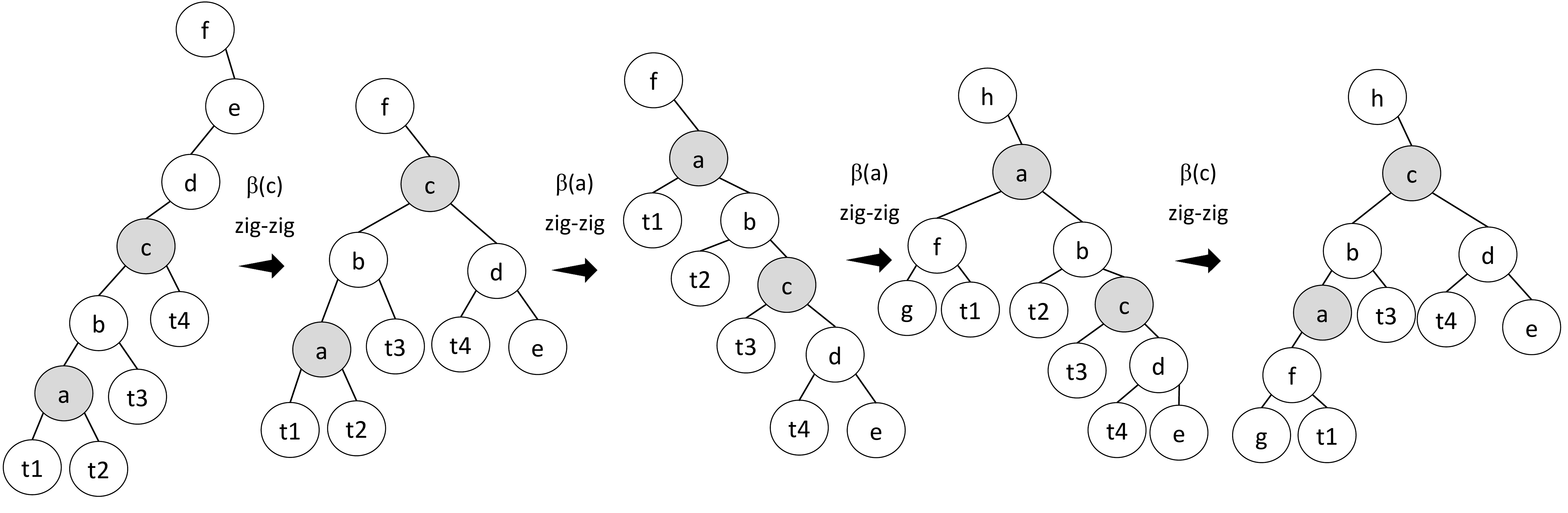}}
    \caption{Loop with a grandparent (Case 1.b.i.B.I): $c=b.p=a.p.p(t), \beta_t(c)$ is a zig-zig and $b$ is a \textit{carried child}; $\beta(a)$ is a zig-zig; $a = c.p.p(t+1)$, so $\beta_{t+1}(a)$ has priority, and it is a zig-zig; $\beta_{t+1}(c)$ is a zig-zig. Even though $a,c$ rotate with each other multiple times, both advance towards their destinations.}
    \label{fig:loopzigzig2.1}
    \end{center}
\end{figure}

 \begin{figure}[!ht]
    \begin{center}
    {\includegraphics[width=0.8\columnwidth]{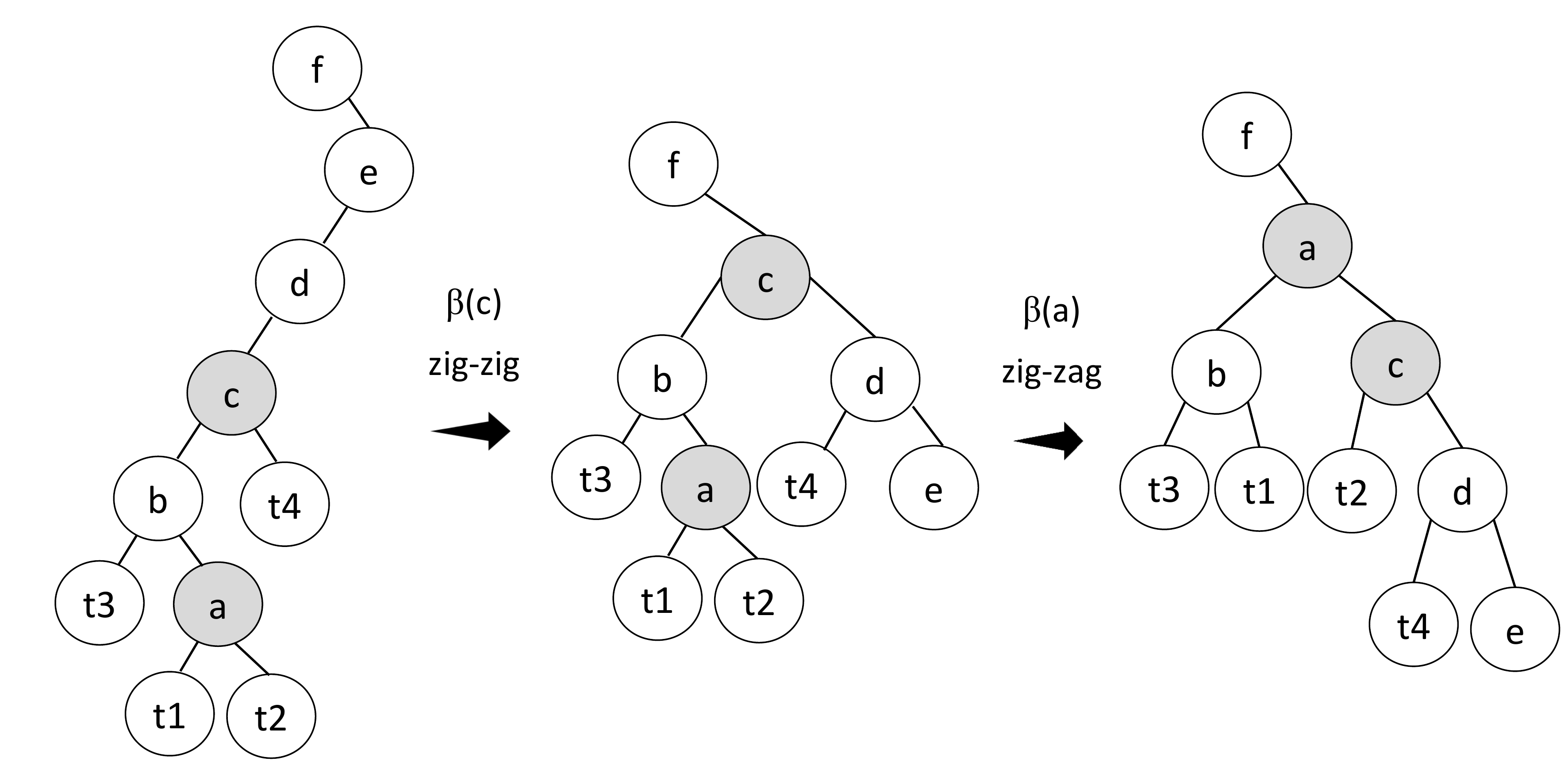}}
    \caption{Loop with a grandparent (Case 1.b.i.B.II): $c=b.p=a.p.p(t), \beta_t(c)$ is a zig-zig and $b$ is a \textit{carried child}; $\beta_t(a)$ is a zig-zag $a = c.p(t+1)$. Thus, any further rotations fit into Case 1(a) (loop with a parent)}
    \label{fig:loopzigzig2.2}
    \end{center}
\end{figure}

\begin{figure}
	\centering
	\begin{subfigure}{0.4\textwidth} 
		\includegraphics[width=\textwidth]{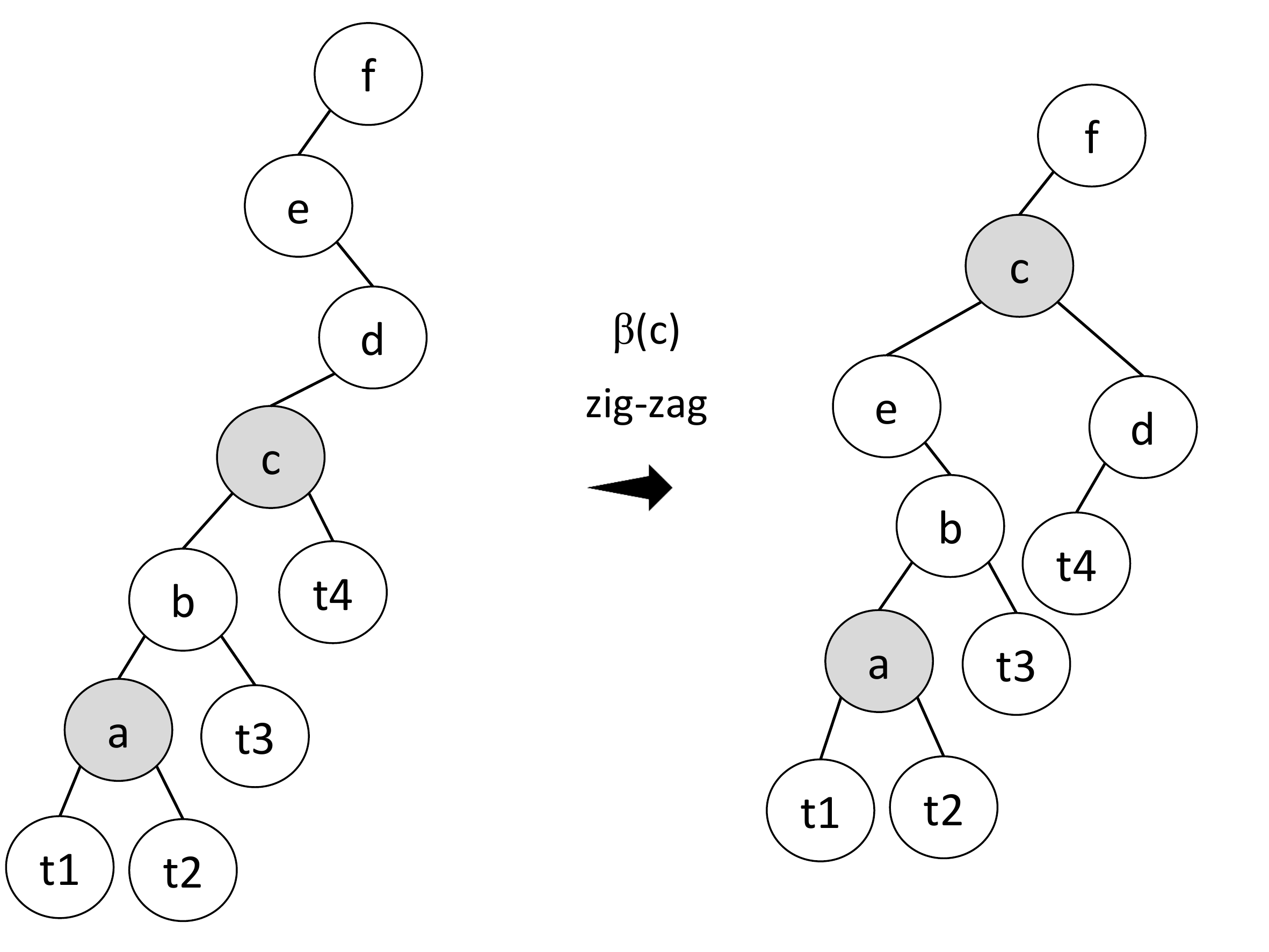}\caption{}
		\label{fig:zigzag1}
	\end{subfigure}
	\vspace{1em} 
	\begin{subfigure}{0.4\textwidth} 
		\includegraphics[width=\textwidth]{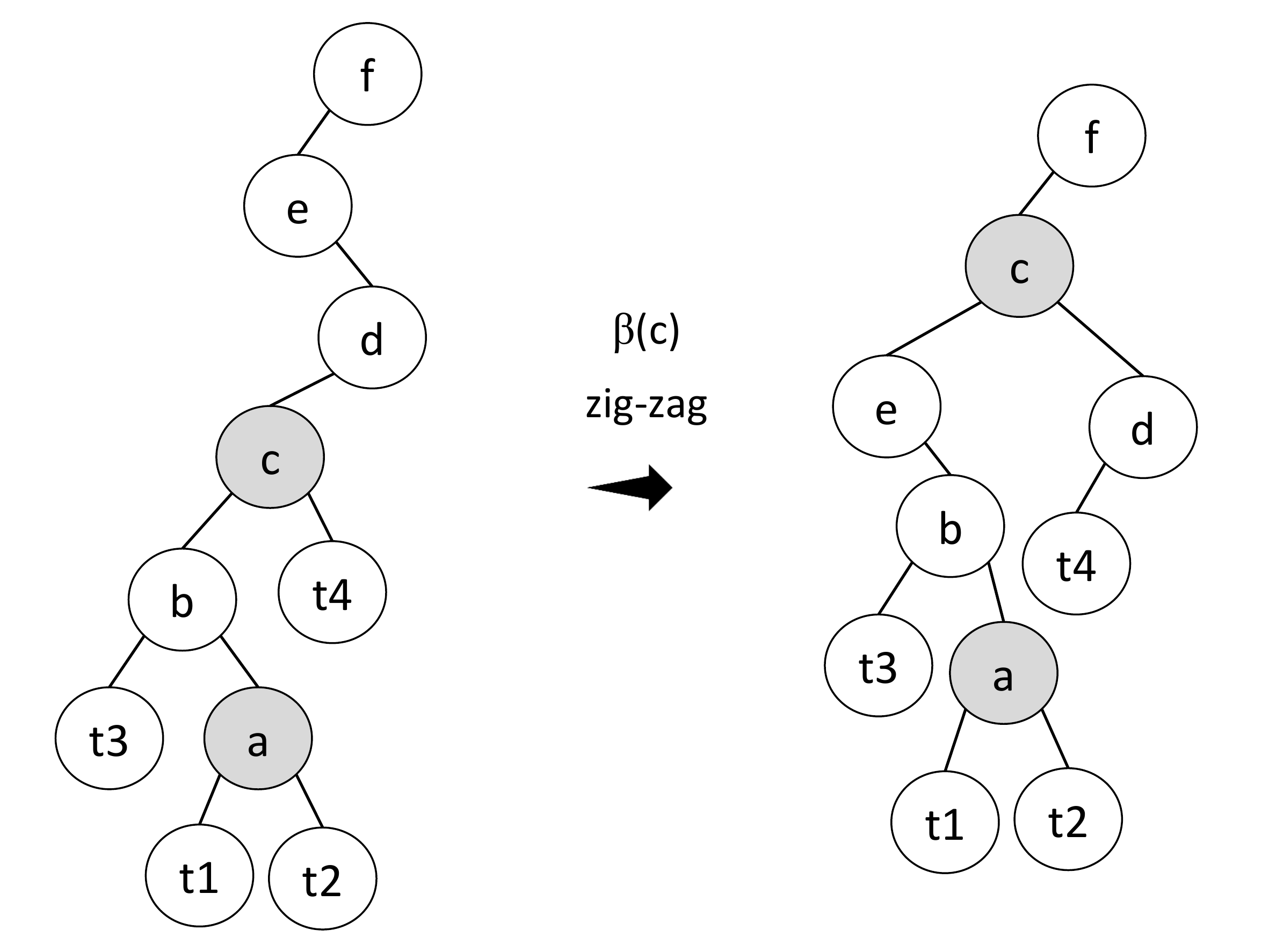}\caption{}
		\label{fig:zigzag2}
	\end{subfigure}	
	\begin{subfigure}{0.4\textwidth} 
		\includegraphics[width=\textwidth]{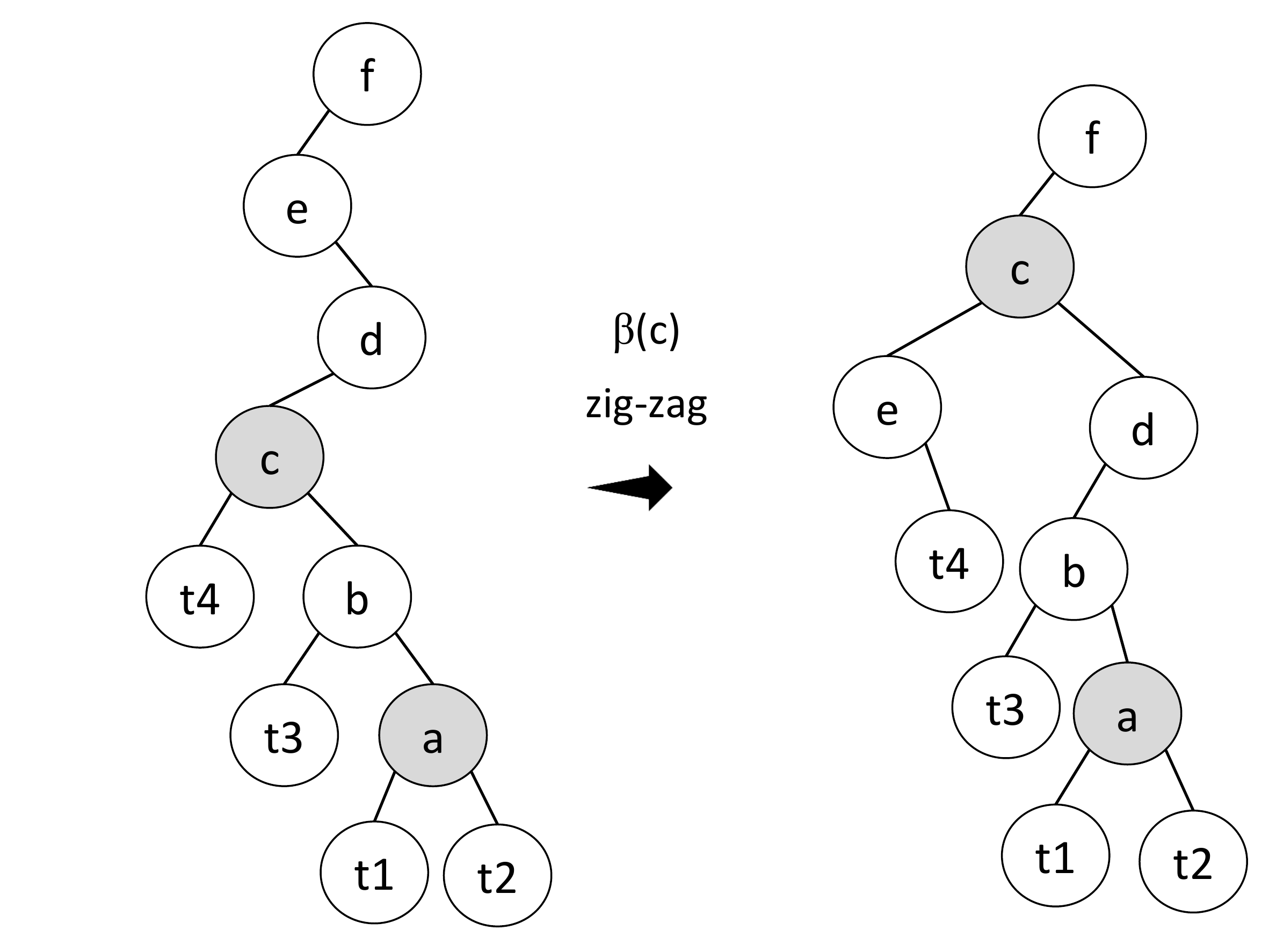}\caption{}
		\label{fig:zigzag3}
	\end{subfigure}	
	\begin{subfigure}{0.4\textwidth} 
		\includegraphics[width=\textwidth]{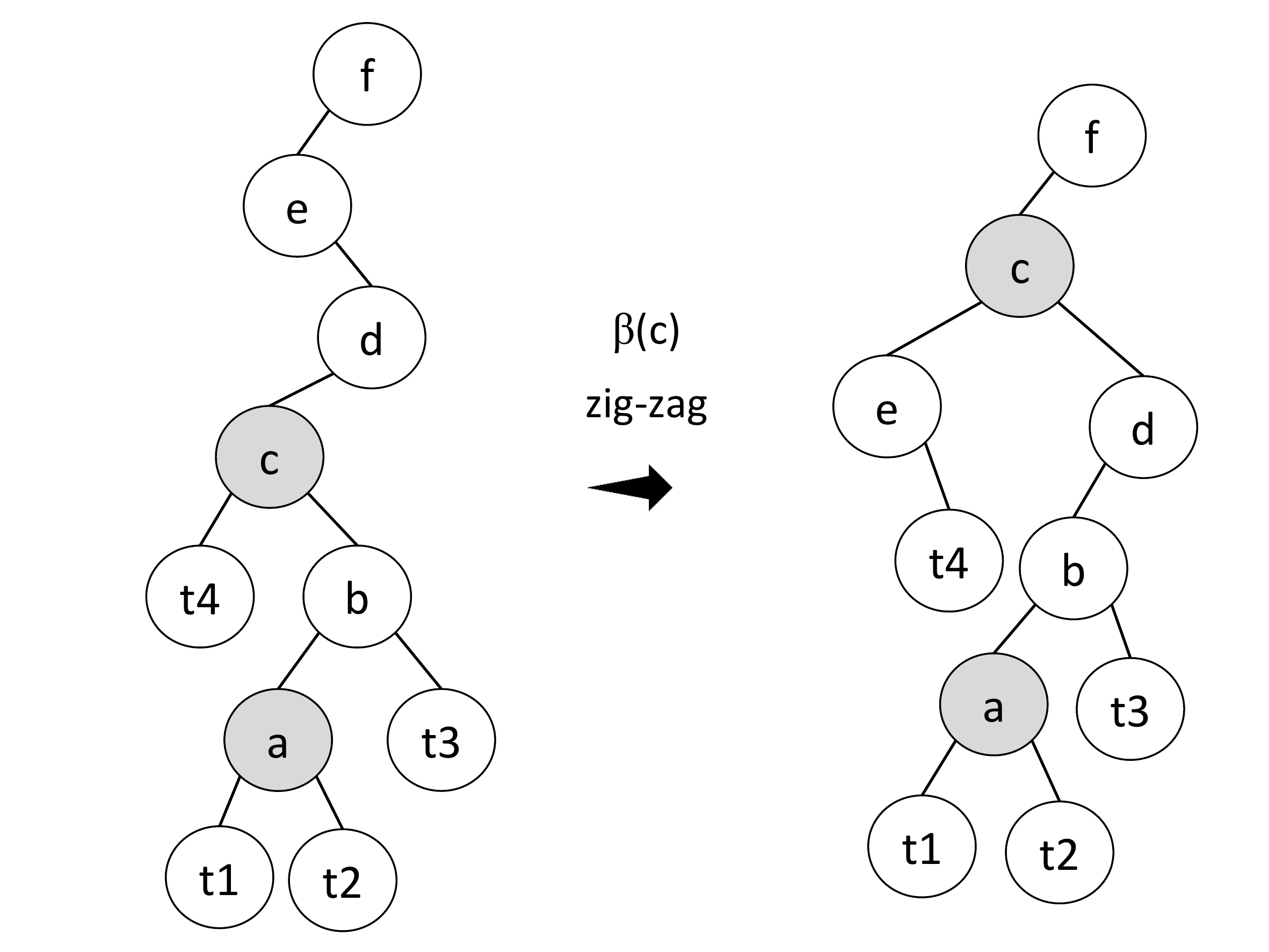}\caption{}
		\label{fig:zigzag4}
	\end{subfigure}
    \caption{Loop with a grandparent (Case 1.b.ii): $c=b.p=a.p.p(t), \beta_t(c)$ is a zig-zag, which results in $d_{t+1}(a,c) = 3$ in all four cases.}
    \label{fig:loopzigzag}
\end{figure}

\begin{figure}
	\begin{subfigure}{0.4\textwidth} 
		\includegraphics[width=\textwidth]{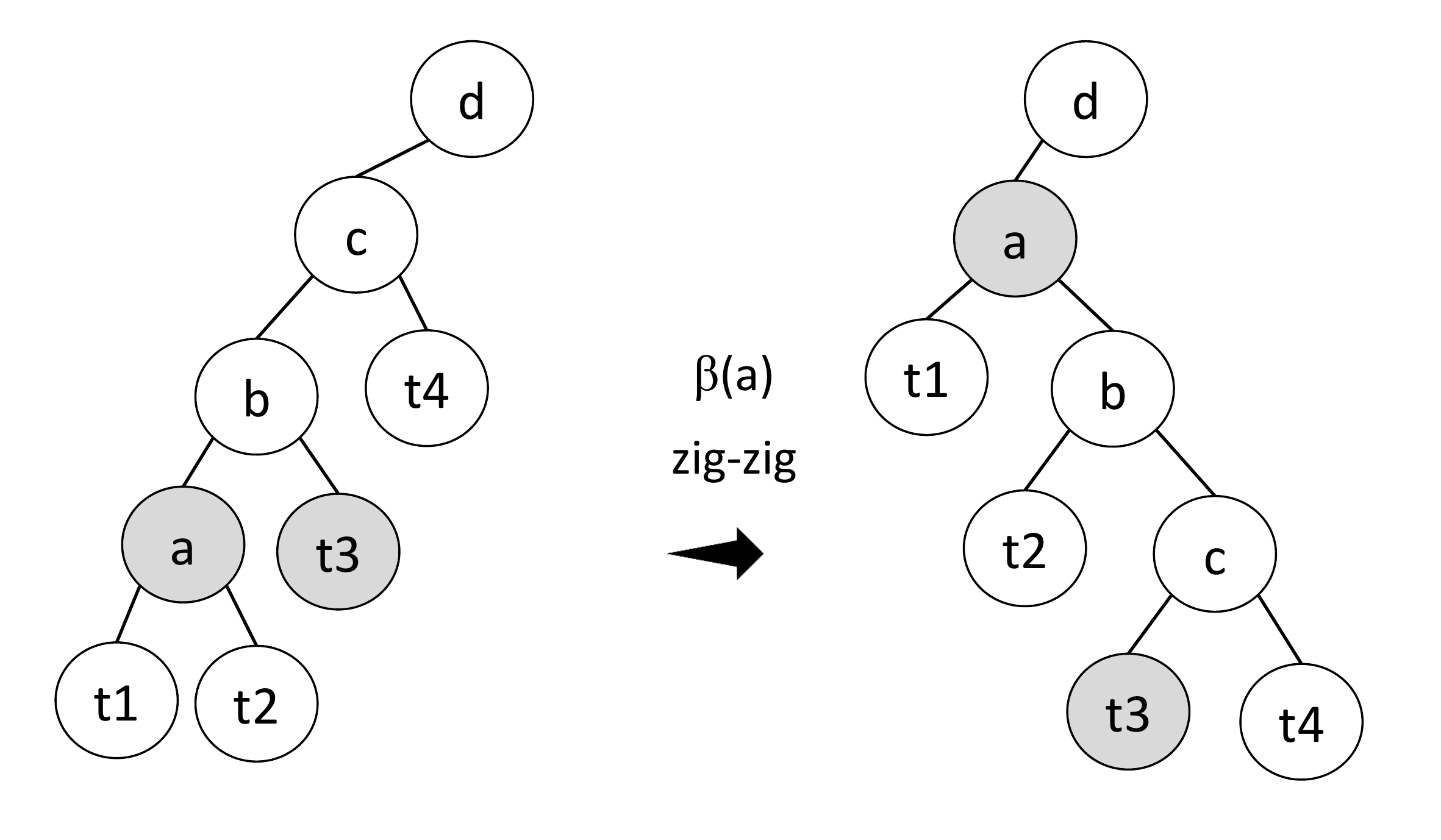}\caption{}
		\label{fig:azigzig}
	\end{subfigure}
	\vspace{1em} 
	\begin{subfigure}{0.4\textwidth} 
		\includegraphics[width=\textwidth]{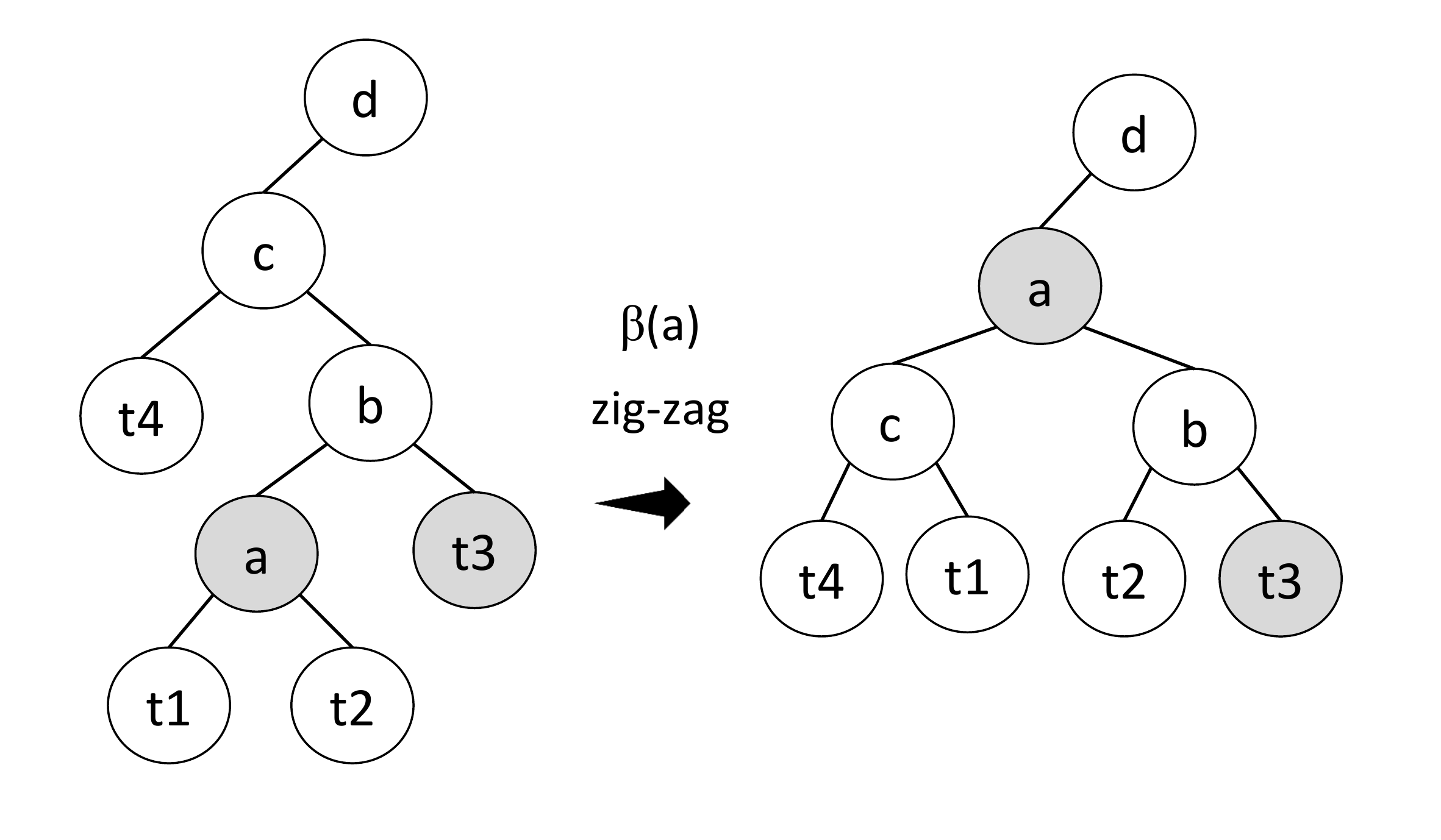}\caption{}
		\label{fig:azigzag}
	\end{subfigure}	
    \caption{Loop with a sibling (Case 1.c): $b=a.p=t3.p$ (a) Case (i): $\beta_t(a)$ is a zig-zig: $d_{t+1}(a,t3) = 3$; (b) Case (ii) $\beta_t(a)$ is a zig-zag: $a = t3.p.p(t+1)$ and $\beta_t(t3)$ fits into Case 1(b) (loop with a grandparent).}
    \label{fig:loopsiblings}
\end{figure}

\end{proof}

\begin{defn}\label{def:inconsistency}
\textbf{Buffer inconsistency: } Consider a SplayNet $\mathcal{T}$, a set of concurrent splay requests $\mathcal{F}$ and two splay request $\mathcal{S}(a,d_a), \mathcal{S}(b,d_b) \in \mathcal{F}$. Consider that in time-slot $\tau$ nodes $a$ and $b$ request a rotation, e. g., $\beta_\tau(a)$ and $\beta_\tau(b)$. We say that there is a buffer inconsistency if there are two nodes $i,j$ that inserted these requests in reverse order into their buffers, i.e., $\mathcal{B}_i=[\ldots,\beta(a),\ldots,\beta(b), \ldots]$ and $\mathcal{B}_j=[\ldots,\beta(b),\ldots,\beta(a), \ldots]$ in time-slot some $\tau'>\tau$. 
\end{defn}

\begin{lemma}\label{claim:noDeadlocks}
\textit{Priority routine} prevents deadlocks from occurring in the network. 
\end{lemma}
\begin{proof}
A deadlock may occur in the network if there is a pair of inconsistent buffers (Definition \ref{def:inconsistency}) or a circular wait. We analyze these two cases below.
\begin{enumerate}
    \item \textbf{Inconsistent buffers}: Consider four nodes $i,j,a,b$, and two inconsistent buffers: $\mathcal{B}_i=[\ldots,\beta(a),\ldots,\beta(b), \ldots]$ and $\mathcal{B}_j=[\ldots,\beta(b),\ldots,\beta(a), \ldots]$. W.l.g., assume $ID(a) < ID(b)$. Since $\beta(a)<\beta(b)$, applying priority routine on $\mathcal{B}_i$, there are three possibilities:
    \begin{enumerate}
        \item\label{case1.1} $a,b$ are siblings ($a.p=b.p$);
        \item\label{case1.2} $a,b$ are cousins ($a.p.p=b.p.p$);
        \item\label{case1.3} $b$ is descendant of $a$ ($a=b.p \lor a.p=b.p.p \lor a=b.p.p \lor a = b.p.p.p$);
    \end{enumerate}    
    Since $\beta(b)<\beta(a)$ in $\mathcal{B}_j$, we know that $b$ is ancestor of $a$, otherwise, priority routine would reverse the order. Combining the views at $i$ and $j$, we realize that case (\ref{case1.1}) is not possible, since $a.p\neq b.p$, case (\ref{case1.2}) is not possible since $a.p.p\neq b.p.p$, and case (\ref{case1.3}) is not possible since $b$ cannot be descendant and ancestor of $a$ at the same time, since such a relationship would create a cycle in the SplayNet, which is a tree. Therefore, there are no inconsistent buffers when priority routine is used to order the rotation requests.
    \item \textbf{Circular wait}: Assume there are $i\geq3$ rotation requests in a circular wait among nodes $a_1 \to a_2 \to \dots \to a_i \to a_1$ in round $t \in \mathcal{R}$. This implies that there is a cycle in the SplayNet, which is a contradiction, since concurrent rotations do not create cycles, even if they belong to different rounds (Lemma~\ref{claim:noLoops}). 
\end{enumerate}
\end{proof}

\section{Round Length}\label{subsec:round}

The communication routines for a rotation $\beta(u)$, described in Section~\ref{subsec:distributed}, must send messages and receive acknowledgements from nodes up to three hops away from $u$. If there were no concurrent rotations, this would take a constant number of time-slots in the synchronous communication model. However, due to concurrent rotation requests, the communication routines of $\beta(u)$ are not necessarily consecutive. Lemma \ref{claim:round} provides an upper bound on the number of time-slots a rotation takes, considering concurrency.

\begin{lemma}\label{claim:round}
Consider a set of concurrent splay requests $\mathcal{F}, |\mathcal{F}|=m$ in super-round $\mathcal{R}$. The length of a round $|t|, t\in \mathcal{R}$, defined as the number of time-slots between a node requesting a rotation and completing it, is $|t|=O(\log{m})$. 
\end{lemma}
\begin{proof}
Consider a rotation request $\beta(u)$, issued by the source node $u$ of some splay request $\mathcal{S}(u,d_u)\in \mathcal{F}$ in time-slot $\tau_1$ in round $t \in \mathcal{R}$.
We know that no requests are lost due to failures and no infinite loops or deadlocks occur (Claims \ref{claim:noLoops} and \ref{claim:noDeadlocks}). As discussed above, the communication routines of a rotation take a constant number of time-slots ($k=$\ref{item:lastStep}). However, each of these communication routines, performed by a node $w$ may be delayed due to other rotation requests in the buffer $\mathcal{B}_w$ with higher priority than $\beta(u)$. Since the concurrency lock (Section \ref{subsec:lock}) of a rotation request affects a constant number of nodes ($k'$) and \textit{priority routine} gives priority to rotation requests by non-decreasing round sequence number, the total number of requests with higher priority than $\beta(u)$ is constant ($k''$). Each of these $k''$ requests might have to wait for another $k''$ requests with higher priority than their own, and so on, resulting in a dependency tree of degree $k''$ and height $\log_{k''}{m}$.  Consequently, from the time-slot at which source node $u$ sends the first rotation request message to its parent, it takes at most $k \log_{k''}{m} = O(\log{m})$ time-slots for the rotation to complete. Afterwards, $u$ may start the next round, with sequence number $t+1 \in \mathcal{R}$.

For example, for the zig-zig rotation request $\beta(u)$ in Figure~\ref{fig:rotations}, $k'=k''=7$. Let's add the following notation: $t_i$ is the root of sub-tree $T_i$, $x=w.p$, $y=w.p.p$, $t_5=x.r$, $t_6=y.r$. The nodes that need to be locked for $\beta(u)$ are: $t_1, t_2, u, t_3, v, w, x$. $\beta(u)$ has maximum priority at buffers of $t_1, t_2$, and $t_3$. The rest of the buffers are ordered as follows: $\mathcal{B}_v=[w,v,t_4,u]$, $\mathcal{B}_w=[x,w,t_5,v,t_4,u]$, and $\mathcal{B}_x=[y,x,t_6,w, t_5, v, t_4,u]$. Therefore, $u$ will wait for $7$ rotations to complete before it acquires all the necessary locks for $\beta(u)$. Each of them might have to wait for another $k''$ requests, and so on.
\end{proof}

\section{Amortized analysis}
\label{subsec:amortized}

In the previous sections, we allowed rounds to overlap, in order to increase concurrency. For the formal analysis, we assume that rounds are (globally) synchronized, i.e., every node performs at most one rotation per round, until reaching its objective. We measure the total work of SplayNet in terms of maximum number of rounds each source-destination pair takes to reach the objective of the splay request. The length of the super round is then the maximum number of rounds needed for all splay requests to reach their objectives concurrently. We compute the amortized cost in terms of rounds, and in the end multiply it by the maximum round length in time-slots, using the upper bound from Lemma~\ref{claim:round}.
        
    Consider a super-round $\mathcal{R}={t_0,..., t_{|\mathcal{R}|}}$, a sequence of SplayNet instances $\mathcal{T} = {\mathcal{T}_{0},..., \mathcal{T}_{{|\mathcal{R}|}}}$ on $n$ nodes,  and a set of concurrent splay requests $\mathcal{F} \in \mathcal{T}$, $|\mathcal{F}|=m$. We analyze the amortized cost using the potential method from \cite{Tarjan:1985p99}.  
    
    \begin{defn}
    Consider that each node $u$ in the SplayNet instance $\mathcal{T}_{i}, t_i \in \mathcal{R}$ is assigned a size $s_i(u)$, which represents the number of nodes in the subtree of $u$ including $u$. We define the \textbf{rank} $r_i(u)$ of a node $u$ as the logarithm in base 2 of the size of $u$, i. e., $r_i(u) = \log_2(s_i(u))$. We define the total SplayNet rank $r(\mathcal{T}_{i})$ as the sum of the ranks of all nodes in $\mathcal{T}_{i}$. Note that the maximum size and rank of a node is $n$ and $\log_2{n}$, respectively. \label{def:totalRank}
    \end{defn} 
     
    We employ a potential function argument and we use the abstraction of ``cyber-dollars'' to pay for the work for splaying a source node $s$ in $\mathcal{T}_{i}$, assuming that one rotation costs one cyber-dollar. Thus, a zig costs one cyber-dollar, while zig-zig and zig-zag cost two cyber-dollars. 

    \begin{defn}\label{def:costSplay}
    We define \textit{cost} $p_j$, in cyber-dollars, of a splay request $\mathcal{S}_j(s,d_s) \in \mathcal{F}$, as the number of rotations that nodes $s$ and $d_s$ need to perform to reach their objective ($s = d_s.p(\tau_j)$, or, $d_{t_k}(s,d_s)=1$ in some time-slot $\tau_k \in \mathcal{R}$. Splaying $\mathcal{S}_j(s,d_d)$ consists of at most $d_{t_i}(s,d_s)/2$ rotations of type zig-zig or zig-zag, plus two zig rotations, if $d_i(s,\alpha(s,d_s))$ and $d_i(d_s,\alpha(s,d_s))$ are odd, for some $t_i\leq t_k \in \mathcal{R}$.
    \begin{equation}        
    p_j \leq \frac{1}{2} \max_{\tau_i \in \mathcal{R}}{d_i(s,d_s)} + 2, \forall \mathcal{S}_j(s,d_s) \in \mathcal{F}, \label{eq:costSplay} 
    \end{equation}
    
    \noindent where $d_i(s,d_s)$ is the hop distance between $s$ and $d_s$ in time-slot $\tau_i$.        
    \end{defn}

    \begin{defn}\label{def:maxSplayDistance}
    We define the maximum splay distance $\mathcal{D}$ in super-round $\mathcal{R}$ as: 
    \begin{equation}
    \mathcal{D} = \max_{\mathcal{S}_j(s,d) \in \mathcal{F}}{p_j} \label{def:splayDistance}
    \end{equation}
    \noindent where $p_j$ is the cost of splay $\mathcal{S}_j \in \mathcal{F}$ (Definition(\ref{def:costSplay}). The maximum splay distance $\mathcal{D}$ is also the upper bound on the number of rounds in a super-round $\mathcal{R}$.
    \end{defn}

The amortized analysis is the average performance of each operation in the worst case~\cite{Cormen:2001}. In SplayNet scenario, the amortized cost can be described as the average cost per operation for a given sequence $\mathcal{F}$ of communication requests. The potential method defines a function that maps a data structure onto a real-valued, non-negative "potential". The potential stored in the data structure may be used to pay for future operations. In the potential method, the amortized cost $\hat{c}_i$ of an operation $i$ is the actual cost $c_i$ plus the increase in potential $\delta$ due to the operation, where $\delta = \phi(D_{i}) - \phi(D_{i-1})$. This gives us:

\begin{equation}\label{eq:potential_individual}
\hat{c_{i}} = c_{i} + \phi(D_{i}) - \phi(D_{i-1})
\end{equation}

This means that the amortized cost and potential function must be defined in order to always maintain such equivalence. This equivalence implies that if the actual cost of and operation is less than the amortized cost, the potential is increased, and if the cost of an operation is greater than its amortized cost, the potential is decreased. By equation \ref{eq:potential_individual}, we can derive the total amortized cost given the actual costs: 
\begin{equation}
\sum^{n}_{i=1}\hat{c_{i}} = \sum^{n}_{i=1}c_{i} + \phi(D_{n}) - \phi(D_{0})
\end{equation}

After performing a splay, one of three cases applies: (1) If the payment is equal to the splaying work, then we use it all to pay for the splaying; (2) If the payment is greater than the splaying work, we deposit the excess in the accounts of several nodes, i. e., we increase the potential; (3) If the payment is less than the splaying work, we make withdrawals from the accounts of several nodes to cover the deficiency, i. e., we decrease the potential. Moreover, an \textit{invariant} is maintained: before and after a splaying, each node $u$ of $\mathcal{T}_{i}$ has $r(u)$ cyber-dollars. To preserve this invariant after a splaying, we must pay the splay work plus the total change in $r(\mathcal{T}_{i})$.

    \begin{lemma}\label{lemma:variationperround} Consider a SplayNet instance $\mathcal{T}_{t_i}$ with $N$ nodes in round $t_i \in \mathcal{R}$, a set of concurrent splay requests $\mathcal{F} \in \mathcal{T}$, $|\mathcal{F}|=m$, and $\leq m$ concurrent rotations in round $t_i$, each belonging to one of the $m$ concurrent splay requests $\mathcal{S}_j \in \mathcal{F}$ in super-round $\mathcal{R}$.\footnote{Note that there can occur $<m$ rotations in round $t_i$, in case some of the splay requests in $\mathcal{F}$ have been completed before round $t_i$.} We denote the rank of a node $u$ before and after a rotation as $r(u)$ and $r'(u)$, respectively.
    Let $\delta_i^j$ be the total rank variation in $r(\mathcal{T}_{t_i})$ caused by all the rotations $\beta_{t_i}(u) \mid u = (s_j \oplus d_j) \in \mathcal{S}_j$ in round $t_i$. We have that:
    
    \begin{itemize}
        \item $\delta_i \leq 3(r'(u)-r(u)) - 2$, if the rotation is a zig-zig or zig-zag;
        \item $\delta_i \leq 3(r'(u)-r(u))$, if the rotation is a zig.
    \end{itemize}
     \end{lemma}

     \begin{proof}
        Figure \ref{fig:rotations} presents the three rotation types used in SplayNet. Let's consider $\delta_i$ the the total rank variation in $r(\mathcal{T}_{t_i})$ in round $t_i$ caused by each rotation. Consider the following inequality:
        \begin{equation}\label{ineq:log}
            \log a + \log b \leq 2\log c - 2, \text{if a>0, b> 0 and c>a+b}
        \end{equation}
        \begin{itemize}
            \item \textbf{zig-zig:} Let's consider the nodes in Figure \ref{fig:rotations}. After a zig-zig rotation, only the rank of nodes $u$, $v$ and $w$ change. In addition, $r'(u) = r(w), r'(v)\leq r'(u)$, and $r(v)\geq r(u)$. Thus,\\
            \(\delta_i = r'(w) + r'(v) + r'(u) - r(w) - r(v) - r(u) \\
            \tab \leq r'(v) + r'(w) - r(u) - r(v)\\
            \tab \leq r'(u) + r'(w) - 2r(u).
            \) \\
            Since $s(u) + s'(w) \leq s'(u)$, by inequality \ref{ineq:log}, $r(u)+r'(w)\leq 2r'(u)-2$, and $r'(w)\leq 2r'(u)-r(u)-2$. Thus,\\
            \(\delta_i \leq r'(u) + (2r'(u)-r(u)-2)-2r(u)\\
            \tab \leq 3(r'(u)-r(u))-2.
            \) 
            
            \item \textbf{zig-zag:} Like in zig-zig, only the ranks of $u$, $v$ and $z$ change in a zig-zag. Also, $r'(u)=r(w)$ and $r(u)\leq r(v)$. Thus,\\
            \(\delta_i = r'(w) + r'(v) + r'(u) - r(w) - r(v) - r(u) \\
            \tab \leq r'(v) + r'(w) - r(u) - r(v)\\
            \tab \leq r'(v) + r'(w) - 2r(u).
            \) \\
            Since $s'(v) + s'(w) \leq s'(u)$, by inequality \ref{ineq:log}, $r'(v)+r'(w)\leq 2r'(u)-2$, and. Thus,\\
            \(\delta_i \leq 2r'(u) -2 -2r(u)\\
            \tab \leq 3(r'(u)-r(u))-2.
            \)  
            
            \item \textbf{zig:} Since a zig only involves two nodes, $u$ and its parent $v$, only the ranks of $u$ and $v$ change. In addition, $r'(v)\leq r(v)$ and $r'(u)\geq r(u)$. Thus,\\
            \(\delta_i = r'(v) + r'(u) - r(v) - r(u) \\
            \tab \leq r'(u) - r(u)\\
            \tab \leq 3(r'(u) - r(u)).
            \)          
        \end{itemize}
           
    \end{proof}

    Now, we can bound the total variation of $r(\mathcal{T})$ caused by $m$ concurrent splay requests in super-round $\mathcal{R}$.

    \begin{lemma}\label{lemma:totalvariation}
    Given a SplayNet $\mathcal{T}$ on $n$ nodes, and a set of concurrent splay requests $\mathcal{F}$, $|\mathcal{F}|=m$ in super-round $\mathcal{R}$. Let $\Delta$ be the total rank variation (potential change) in $\mathcal{T}$ caused by all rotations in $\mathcal{F}$. We have that 
    \begin{equation}
    \Delta \leq -2\sum\limits_{i=1}^D{|\mathcal{S}_i|} + O(m\log{n}),\nonumber
    \end{equation}
    \noindent where $\mathcal{D}$ is the maximum splay distance (Definition \ref{def:maxSplayDistance}), and $\mathcal{S}_i$ is the set of nodes originating a rotation request in round $t_i \in \mathcal{R}$.
    \end{lemma}
    \begin{proof}
    Let $\mathcal{S}_{\mathcal{F}}$ denote the set of source nodes in $\mathcal{F}$ and $\mathcal{D}_{\mathcal{F}}$ the set of destination nodes. Moreover, for each splay request $\mathcal{S}(s_k,d_k) \in \mathcal{F}$, define two time-slots $\tau'_k$ and $\tau''_k$, such that $s_k = \alpha_{\tau'_k}(s_k,d_k)$ and $s_k=d_k.p(t''_k)$.
    
    Using Lemma \ref{lemma:variationperround} and summing over all rounds in $\mathcal{R}$, we have that:
    \begin{eqnarray}
    \Delta &\leq & \sum\limits_{i=1}^{\mathcal{D}} \sum\limits_{k=1}^{|\mathcal{S}_i|}{\delta_i(u_k)}
           \leq  \sum\limits_{i=1}^{\mathcal{D}} \sum\limits_{k=1}^{|\mathcal{S}_i|}{(3(r_{i}(u_k)-r_{i-1}(u_k))-2)+4}\nonumber\\
           &\leq & 4-2\sum\limits_{i=1}^{\mathcal{D}}{|\mathcal{S}_i|}+\sum\limits_{u_k\in\mathcal{S}_{\mathcal{F}}}\sum\limits_{i=1}^{t'_k}{3(r_{i}(u_k)-r_{i-1}(u_k))} + \sum\limits_{u_k\in\mathcal{D}_{\mathcal{F}}}\sum\limits_{i=t'_k}^{t''_k}{3(r_{i}(u_k)-r_{i-1}(u_k))}\nonumber\\
           &\leq & 4-2\sum\limits_{i=1}^{\mathcal{D}}{|\mathcal{S}_i|}+3\sum\limits_{u_k\in\mathcal{S}_{\mathcal{F}}}{(r_{t'_{k}}(u_k)-r_{0}(u_k))} + 3\sum\limits_{u_k\in\mathcal{D}_{\mathcal{F}}}{(r_{t''_{k}}(u_k)-r_{t'_{k}}(u_k))}\nonumber\\
           &\leq & 4-2\sum\limits_{i=1}^{\mathcal{D}}{|\mathcal{S}_i|}+O(m\log{n}), \nonumber
    \end{eqnarray}
    \noindent where $+4$ accounts for the two possible zig rotations in each (double) splay. 
    \end{proof}
           
    
    \begin{theorem}\label{thm:totalcost}
    Given a SplayNet $\mathcal{T}$ on $n$ nodes, and a set of concurrent splay requests $\mathcal{F}$, $|\mathcal{F}|=m$ in super-round $\mathcal{R}$. The amortized average cost, per splay request, is $\text{Cost($\mathcal{T},\mathcal{R},\mathcal{F}$)} = O(\log{n})$ rounds, which takes $O(\log{n}\log{m})$ time-slots in total.
    \end{theorem}
    \begin{proof}
    By Lemma \ref{lemma:totalvariation}, the payment of $\Delta$ cyber-dollars is sufficient both to maintain the invariant and pay for the cost of the $m$ concurrent splay requests in a splay tree on $n$ nodes. Using Definition (\ref{def:costSplay}) of the cost of a splay as the total number of rotations performed to execute it, let $c_m$ denote the total actual cost of $m$ splays in number of rotations. The total cost to perform $m$ concurrent splay requests is: 
    \begin{eqnarray}
    \text{Cost($\mathcal{T},\mathcal{R},\mathcal{F}$)} & \leq & c_m + \Delta
     \leq  2\sum\limits_{i=1}^{\mathcal{D}}{|\mathcal{S}_i|}-2\sum\limits_{i=1}^{\mathcal{D}}{|\mathcal{S}_i|}+O(m\log{n})+4 
    =O(m\log{n}) \text{ rounds},\nonumber
    \end{eqnarray}
        
    \noindent which gives an amortized average cost, per splay request, in number of rounds, of $O(\log{n})$. Multiplying it by the upper bound on the length of a round, given in Lemma~\ref{claim:round}, the total duration, in time-slots, of the super-round $\mathcal{R}$ is $O(\log{n}\log{m})$.
\end{proof}

     \begin{theorem}\label{thm:finalEntropies}
    Given a SplayNet $\mathcal{T}$ on $n$ nodes, and a set of concurrent splay requests $\mathcal{F}$, $|\mathcal{F}|=m$ in super-round $\mathcal{R}$. Let $H(\hat{X})$ and $H(\hat{Y})$ denote the empirical entropies (Definition \ref{def:entropy}) of the source and destination nodes in $\mathcal{F}$, respectively. The average amortized cost to serve $\mathcal{F}$ is $O((H(\hat{X})+H(\hat{Y}))\log{m})$.
    \end{theorem}
    \begin{proof}   
    Let $\mathcal{S}_{\mathcal{F}}$ denote the set of source nodes in $\mathcal{F}$ and $\mathcal{D}_{\mathcal{F}}$ the set of destination nodes. Moreover, for each splay request $\mathcal{S}(s_k,d_k) \in \mathcal{F}$, define two time-slots $\tau'_k$ and $\tau''_k$, such that $s_k = \alpha_{\tau'_k}(s_k,d_k)$ and $s_k=d_k.p(t''_k)$.
    
    Let $sf(u_k)$ and $df(u_k)$ be the number of times $u_k$ appears as source and destination in $\mathcal{F}$, respectively. Let's define two weights of $u_k$: $sw(u_k)= \frac{sf(u_k)}{m}$ and $dw(u_k)= \frac{df(u_k)}{m}$. Moreover, let's define two types of size: $ss(u_k)=\sum\limits_{u_k\in T_k}{sw(u_k)}$, where $T_k$ is the sub-tree rooted at $u_k$ (and analogously, $ds(u_k)$), and ranks $sr(u_k)=\log_2{ss(u_k)}$ (and $dr(u_k)$). Using the fact that $ss(u_k)\geq sf(u_k)/m, u_k \in \mathcal{S}_{\mathcal{F}}$ and $s(u_k)\leq 1, u_k \notin \mathcal{S}_{\mathcal{F}}$ (same for $ds(u_k)$), we have that:
    
    \begin{eqnarray}
    \text{Cost($\mathcal{T},\mathcal{R},\mathcal{F}$)} &\leq & 4+3\sum\limits_{u_k\in\mathcal{S}_{\mathcal{F}}}{(r_{t'_{k}}(u_k)-r_{0}(u_k))}
            + 3\sum\limits_{u_k\in\mathcal{D}_{\mathcal{F}}}{(r_{t''_{k}}(u_k)-r_{t'_{k}}(u_k))}\nonumber\\
           &\leq & 4+3\sum\limits_{u_k\in\mathcal{S}_{\mathcal{F}}}{sf(u_k)\log_2{\frac{m}{sf(u_k)}}}
           + 3\sum\limits_{u_k\in\mathcal{D}_{\mathcal{F}}}{df(u_k)\log_2{\frac{m}{df(u_k)}}}\nonumber\\
           &=& O((H(\hat{X})+H(\hat{Y})) \text{ rounds}.\nonumber
    \end{eqnarray}
    
    \noindent Multiplying by the length of a round (Lemma \ref{claim:round}), we have:
    \begin{eqnarray}
        \text{Cost($\mathcal{T},\mathcal{R},\mathcal{F}$)} &= &O((H(\hat{X})+H(\hat{Y})\log{m}) \text{ time-slots},\nonumber
    \end{eqnarray}
    \noindent which is a $O(\log{m})$ factor larger than the bound in Theorem~\ref{thm:entropies}, where no concurrency among the splay requests was considered.
    
\end{proof}

\chapter{Next Steps}\label{chap:future}
In this chapter we present the future directions that we intend to follow in this thesis project. They are divided into four categories: algorithm (section \ref{sec:future/algorithm}), analysis (section \ref{sec:future/analysis}), dataset (section \ref{sec:future/dataset}) and experimentation (section \ref{sec:future/experimentation}).

\section{Algorithm}\label{sec:future/algorithm}
The algorithm proposed in this project implements an abrupt trigger for splays in SplayNet. This means that once a source node performs a communication request to a destination node, that source node triggers a splay operation. Once a source node requests a destination node, rotations are performed until the splay is complete. Thus, a communication request and a splay request are equivalent in our model. Is easy to see that such trigger policy is efficient if the communication pattern presents high locality, i. e., a pair of nodes that have communicated in the past are likely to communicate again in the future. 

As a future direction, we intend to propose and evaluate different policies for triggering splay operations, in order to optimize the access time between nodes, given a communication request. One way to do this is by \textit{Lazy Splaying}. Rather than an access request to trigger rotations all the way between a source and a destination, only a few rotations can be performed per request. Thus, if the frequency of communication requests between the same source and destination is high, these nodes are more likely to become neighbors. 

Other Splay policies, considering different parameters, such as concurrency level or the potential function will also be proposed, implemented and analyzed accordingly, based on \cite{Afek2012}, \cite{Afek2014}, \cite{Sikder:skipGraphs} and \cite{Reiter08}. The purpose of new splay trigger policies is to be able to control the number of rotations per splay or the number of communication requests required to trigger a splay request. In this way, pairs of nodes that communicate more frequently are most likely to become closer and stay that way, than pairs of nodes that rarely (or never) communicate, optimizing network concurrency. In addition, we intend to reduce the cost of a splay, by reducing the cost of concurrency in the network, when compared to the abrupt splay trigger.

\section{Analysis}\label{sec:future/analysis}
The proposed concurrent SplayNets is based on conservative locking, in which nodes participating in a rotation lock themselves during the rotation. Since rotations are sorted according to priority criteria, which includes the hierarchy of the source node in the network, some nodes that have a link changed by the rotation do not need be locked. In this way, we have been able to provide a greater degree of concurrency with efficient locking routines. As a future direction, we intend to analyze in detail the impact caused by the locking, as well as the inaccuracies caused by its absence. In the current analysis, the cost of concurrency is $O(\log m)$, where $m$ is the number of concurrent splays. In addition, we want to compare the effect of introducing concurrency among rotations on the network, with the mutual exclusion approach, which allows only one splay at a time. 

The analysis presented in this project is divided into two main directions. First, we analyze the main issues that might occur in SplayNets upon introducing concurrency. Afterward, we present an analysis of the cost, in terms of time, of the rotations and splays. For future directions we intend to expand our analysis, considering the different splay trigger policies cited in section \ref{sec:future/algorithm}. We envision to compare how each of these policies may affect SplayNet performance, given the network and the communication pattern. We also want to consider different analysis techniques, in order to achieve new analytical results. One possible direction is to analyze the dynamic optimality conjecture, proposed in \cite{Sleator:1985:SBS:3828.3835}, which claims that splay trees perform as well as any other binary search tree algorithm up to a constant factor. This is an unproven conjecture, that may be adapted to SplayNets.

\begin{defn}
\textbf{Dynamic Optimality Conjecture} (\cite{Sleator:1985:SBS:3828.3835}): Let $\mathcal{A}$ be any algorithm that carries out each access by traversing the path from the root to the node (let's consider node $x$) containing the accessed item, at a cost of $d(x) + 1$, where $d(x)$ is the depth of $x$, and that between accesses performs an arbitrary number of rotations anywhere in the tree, at a cost of one per rotation. Then, the total time to perform all the accesses by splaying is no more than $O(n)$ plus a constant times the time required by algorithm $\mathcal{A}$. 
\end{defn}

\section{Dataset}\label{sec:future/dataset}

The analysis presented in \cite{Roy:2015} shows that, unlike reported on in literature the majority of traffic in a data center is not rack-local. Instead, almost 60\% of all traffic from Facebook's data center clusters are inter-rack (intra-cluster). Thus, combining SplayNet with ProjecToR can optimize the communication performance between racks inside a data center. Part of the data collected and analyzed in \cite{Roy:2015} is available, and we intend to use this data to perform a more realistic simulation of SplayNet. The data consists of samples from 3 production clusters. Cluster-A is for Database, Cluster-B is for Web servers, and Cluster-C is used as Hadoop servers. All three clusters are in Facebook's Altoona Data Center.

Another potentially interesting traffic trace to consider is the one presented by \cite{Garcia2014100}. The dataset, called CTU-13, is composed of botnet traffic that was captured in the CTU University, Czech Republic, in 2011. This dataset is available for download as a pcap file for the botnet capture only and a labeled NetFlow file \footnote{CTU-13 dataset available in:  http://mcfp.weebly.com/the-ctu-13-dataset-a-labeled-dataset-with-botnet-normal-and-background-traffic.html}.

\section{Experimentation}\label{sec:future/experimentation}
The first step regarding experimentation is under development. We are implementing the proposed algorithm in Sinalgo (\cite{Sinalgo}). The main purpose of this implementation is to verify and validate the proposed algorithm. Sinalgo is a network simulator that focuses on the verification of network algorithms, and abstracts from the underlying layers. It offers a message passing view of the network, which captures well the view of actual network devices, but it does not simulate the different layers of the network stack. Thus, we do not take into account the protocol stack design. This is probably the biggest decision we must make in the future to implement the protocol prototype and compare it to some realistic baselines. At first, we intend to implement a concurrent version of SplayNet in Omnet++ (\cite{omnet}) to compare simulation results with some data center baseline. Omnet++ is an open-source discrete event simulator, designed to support modeling very large networks from reusable model components. Omnet ++ has several features that allow the implementation of SplayNets in a data center scenario.

 The experimental results can be used to compare the cost and benefit of SplayNet in a real scenario, and can provide us additional information to improve the model and analysis. In order to accomplish such task, we can measure the communication latency, the number of splay operations, the message exchange among other parameters. We also intend to use the results to ratify and complement the theoretical analysis.





\end{document}